\documentclass[11pt,reqno,tikz]{amsart}
\usepackage{amsmath,amsthm,amssymb,bm,bbm,dsfont,braket,esint,nicefrac}
\usepackage[mathscr]{eucal}
\usepackage{amsaddr}
\usepackage{upgreek}
\usepackage[left=2cm,right=2cm,top=2cm,bottom=2cm]{geometry}
\usepackage{mathtools}\mathtoolsset{centercolon}
\usepackage[nospace,noadjust]{cite}

\usepackage[shortlabels]{enumitem}
\usepackage{setspace}
\usepackage{graphicx}
\usepackage{verbatim}
\usepackage{float}
\usepackage{placeins}
\usepackage{array}
\usepackage{booktabs}
\usepackage{multicol,multirow,tabularx}
\usepackage{threeparttable}
\usepackage[update,prepend]{epstopdf}
\usepackage{amsfonts,amssymb,dsfont,xfrac,comment}
\usepackage[abs]{overpic}

\makeatletter
\def\adl@drawiv#1#2#3{%
	\hskip0
	\tabcolsep
	\xleaders#3{#2 0\@tempdimb #1{1}#2 0.5\@tempdimb}%
	#2\z@ plus1fil minus1fil\relax
	\hskip0\tabcolsep}
\newcommand{\cdashlinelr}[1]{%
	\noalign{\vskip\aboverulesep
		\global\let\@dashdrawstore\adl@draw
		\global\let\adl@draw\adl@drawiv}
	\cdashline{#1}
	\noalign{\global\let\adl@draw\@dashdrawstore
		\vskip\belowrulesep}}
\makeatother

\usepackage[usenames,dvipsnames]{color}
\usepackage{colortbl}
\usepackage{arydshln}
\usepackage[hidelinks,linktocpage=true]{hyperref}
\hypersetup{
	unicode=false,          
	pdftoolbar=true,        
	pdfmenubar=true,        
	pdffitwindow=false,     
	pdfstartview={FitH},    
	pdftitle={My title},    
	pdfauthor={Author},     
	pdfsubject={Subject},   
	pdfcreator={Creator},   
	pdfproducer={Producer}, 
	pdfkeywords={keyword1} {key2} {key3}, 
	pdfnewwindow=true,      
	colorlinks=true,        
	linkcolor=Red,          
	citecolor=ForestGreen,  
	filecolor=Magenta,      
	urlcolor=BlueViolet,    
}
\usepackage{doi}
\usepackage{url}
\usepackage{caption, subcaption}
\usepackage{enumitem}
\usepackage{shadethm}
\usepackage{tikz,pgf}
\usepackage{witharrows}

\usepackage{cleveref}

\newcommand{\Renyi}{{}R\'{e}nyi{ }}

\makeatletter
\ifx\@NODS\undefined%

\let\mathbb=\mathds
\else%
\fi
\makeatother

\def\d{{\text {\rm d}}}

\def\E{\mathsf{E}}
\def\Z{\mathsf{Z}}

\DeclareMathOperator{\Tr}{Tr}

\DeclareMathOperator{\I}{\mathbf{1}}

\DeclareMathOperator{\e}{\mathrm{e}}

\def\X{\mathsf{X}}

\def\A{\mathsf{A}}
\def\B{\mathsf{B}}
\def\C{\mathsf{C}}

\def\R{\mathsf{R}}

\newcommand{\proj}[1]{\left\{#1\right\}}

\newcommand{\be}{{\mathbf e}}

\newcommand{\pl}{\hspace{.1cm}}

\def\0{{\mathbf{0}}}
\def\1{{\mathbf{1}}}
\def\2{{\mathbf{2}}}
\def\3{{\mathbf{3}}}
\def\4{{\mathbf{4}}}
\def\5{{\mathbf{5}}}
\def\6{{\mathbf{6}}}

\def\7{{\mathbf{7}}}
\def\8{{\mathbf{8}}}
\def\9{{\mathbf{9}}}

\def\be{\begin{equation}}
\def\ee{\end{equation}}
\def\bea{\begin{eqnarray}}
\def\eea{\end{eqnarray}}

\def\eps{\varepsilon}

\newcommand{\id}{\operatorname{id}}

\theoremstyle{plain}
\newshadetheorem{theo}{Theorem}
\newshadetheorem{prop}{Proposition}[section]
\newshadetheorem{lemm}{Lemma}[section]

\newshadetheorem{theorem}{Theorem}[section]

\theoremstyle{definition}
\newtheorem{defn}{Definition} 

\theoremstyle{remark}
\newtheorem{remark}{Remark}[section]

\makeatletter
\newcommand{\opnorm}{\@ifstar\@opnorms\@opnorm}
\newcommand{\@opnorms}[1]{%
	$\left|\mkern-1.5mu\left|\mkern-1.5mu\left|
	#1
	\right|\mkern-1.5mu\right|\mkern-1.5mu\right|$
}
\newcommand{\@opnorm}[2][]{%
	\mathopen{#1|\mkern-1.5mu#1|\mkern-1.5mu#1|}
	#2
	\mathclose{#1|\mkern-1.5mu#1|\mkern-1.5mu#1|}
}
\makeatother

\usepackage{tikz}
\usetikzlibrary{arrows.meta} 
\tikzset{>={Latex[length=4,width=4]}} 
\usetikzlibrary{calc}

\colorlet{mylightblue}{blue!5!white}
\colorlet{mydarkblue}{blue!30!black}
\colorlet{myblue}{blue!50!black}
\colorlet{myred}{red!50!black}
\colorlet{mydarkred}{red!30!black}
\colorlet{mydarkgreen}{green!30!black}

\newcommand{\sh}{\kern-0.08em$^\textbf{\#}$\hspace{-3pt}}
\renewcommand{\b}{\kern-0.06em$\flat$}

\begin{document}

\let\origmaketitle\maketitle
\def\maketitle{
	\begingroup
	\def\uppercasenonmath##1{} 
	\let\MakeUppercase\relax 
	\origmaketitle
	\endgroup
}

\title{\bfseries \Large{ 
Sharp estimates of quantum covering problems \\via a novel trace inequality
}}

\author{ \normalsize 
{Hao-Chung Cheng}$^{1\text{--}5}$,
{Li Gao}$^6$,
{Christoph Hirche}$^7$,
{Hao-Wei Huang}$^8$,
and
{Po-Chieh Liu}$^{1,2}$
}
\address{\small  	
$^1$Department of Electrical Engineering and Graduate Institute of Communication Engineering,\\ National Taiwan University, Taipei 106, Taiwan (R.O.C.)\\
$^2$Department of Mathematics, National Taiwan University\\
$^3$Center for Quantum Science and Engineering,  National Taiwan University\\
$^4$Hon Hai (Foxconn) Quantum Computing Center, New Taipei City 236, Taiwan (R.O.C.)\\
$^5$Physics Division, National Center for Theoretical Sciences, Taipei 10617, Taiwan (R.O.C.)\\
$^6$School of Mathematics and Statistics, Wuhan University, Wuhan, 430072, China\\
$^7$Institute for Information Processing (tnt/L3S), Leibniz Universit\"at Hannover, Germany\\
$^8$Department of Mathematics, National Tsing Hua University, Hsinchu 300, Taiwan (R.O.C.)
}


\date{\today}

\begin{abstract}
In this paper, we prove a novel trace inequality involving two operators.
As applications, we sharpen the one-shot achievability bound on the relative entropy error in a wealth of quantum covering-type problems, such as soft covering, privacy amplification, convex splitting, quantum information decoupling, and quantum channel simulation by removing some dimension-dependent factors.
Moreover, the established one-shot bounds extend to infinite-dimensional separable Hilbert spaces as well.
The proof techniques are based on the recently developed operator layer cake theorem and an operator change-of-variable argument, which are of independent interest.
\end{abstract}

\maketitle
\tableofcontents

\section{Introduction} \label{sec:introduction}

One of the main research topics in quantum information theory and mathematical physics is to provide tight error estimates to information processing tasks or a physical process.
Nonetheless, due to the noncommutative nature of quantum mechanics, many scalar inequalities do not immediately extend to the matrix setting.
Hence, finding trace inequalities or operator inequalities becomes a crucial research direction in matrix analysis and noncommutative analysis as they serve as fundamental tools for a variety of applications in mathematical physics; see, e.g.~\cite{Car09, Lieb_book_2002, Simon_book_2019, Lieb_book_2022}.

In this paper, we establish a novel trace inequality involving two positive operators $A$ and $B$:
\begin{align} \label{eq:goal}
    \Tr\left[ A \left( \log (A+B) - \log B \right) \right] &\leq
	s^s(1-s)^{1-s} \int_0^\infty\Tr\left[ \left( A (B + t \I )^{-1} \right)^{1+s}  \right]\d t
    \\
    &\leq \left(\frac{1-s}{s}\right)^{1-s} \Tr\left[ \left( B^{-\frac{s}{2(1+s)}} A B^{-\frac{s}{2(1+s)}} \right)^{1+s} \right], \quad \forall\, s\in (0,1].
\end{align}
These upper bounds naturally connect to \Renyi divergences, which are frequently used to give one-shot bounds in information theory. The second upper bound can be restated in terms of the sandwiched \Renyi divergence~\cite{MDS+13, WWY14}, 
\begin{align}
    \widetilde{D}_\alpha(\rho\|\sigma) = \frac{1}{\alpha-1}\log\Tr\left[ \left( \sigma^{\frac{1-\alpha}{2\alpha}} \rho \sigma^{\frac{1-\alpha}{2\alpha}} \right)^{\alpha} \right]
\end{align}
which gives asymptotically optimal error exponents for numerous information theoretic problems. In the one-shot setting, one strives for the tightest bound and can use the first inequality which connects to the \Renyi divergence
\begin{align}
    D_\alpha(\rho\|\sigma) = \frac{1}{\alpha-1}\log\left[(\alpha-1)\int_0^\infty \Tr\left[ \left( \rho (\sigma+t\I)^{-1} \right)^{\alpha} \right] \d t \right]\label{Eq:int-renyi}
\end{align}
This variant was introduced in the form of an integral representation in~\cite{hirche2023quantum} and conjectured to take the above form for $\alpha>1$ in~\cite{beigi2025some}, which was recently proven in~\cite{preparation2}. This divergence gives tighter bounds in the sense that, 
\begin{align}
    D_\alpha(\rho\|\sigma) \leq \widetilde{ D}_\alpha(\rho\|\sigma)\quad\forall\alpha>1, \label{Eq:compRenyi}
\end{align}
as shown in~\cite{beigi2025some,preparation2}. Note that the divergence in Equation~\eqref{Eq:int-renyi} is not additive, however it becomes equal to the sandwiched \Renyi divergence in the limit of many copies~\cite{hirche2023quantum}. 

\smallskip

In the next step, we move to applications of the above inequality in Equation~\eqref{eq:goal}. We show that it can be used to sharpen the one-shot error estimate, in terms of a quantum relative entropy criterion or the purified distance, for 
a series of \emph{quantum covering-type problems}, including 
\begin{itemize}
    \item soft covering (Section~\ref{sec:covering}), 
    \item privacy amplification against quantum side information (Section~\ref{sec:PA}),
    \item convex splitting (Section~\ref{sec:splitting}),
    \item catalytic quantum information decoupling (Section~\ref{sec:decoupling}), 
    \item and entanglement-assisted quantum channel simulation (Section~\ref{sec:simulation}).
\end{itemize}
Our bounds improve on the previous results in the literature in three precise ways: 
\begin{enumerate}
    \item First and most notably, our result removes a factor of the spectral size of $\B$ (the number of distinct eigenvalues) in all aforementioned applications.
This factor in the $n$-fold i.i.d. scenario (i.e., $A\leftarrow A^{\otimes n}$, $B\leftarrow B^{\otimes n}$) grows at most as $(n+1)^{\text{dim} \, \mathcal{H}}$, which does not affect the exponential decay rate but can be significant in the one-shot setting.
More importantly, without the dimension-dependent factor, the established error estimates via the trace inequality \eqref{eq:goal} extend to infinite-dimensional separable Hilbert space as well. This is of practical importance as one may not impose the finite-dimension assumption on a quantum eavesdropper in the application of privacy amplification, for example. 
\item Our results are stated using the recent \Renyi divergence in Equation~\eqref{Eq:int-renyi}, which improves the bounds compared to the sandwiched \Renyi divergence by virtue of Equation~\eqref{Eq:compRenyi}. 
\item Finally, a small constant factor improvement is achieved by including the additional constant $c_s=s^s(1-s)^{1-s} \leq 1$. 
\end{enumerate}

The technical ingredient for proving \eqref{eq:goal} is the layer cake theorem recently established  in \cite{preparation}.
We first express the left-hand side of \eqref{eq:goal} in terms of an integral representation via the fundamental theorem of calculus.
Then, by essentially employing only the \emph{scalar} Young inequality, we obtain the desired right-hand side.
We consider such a proof technique to be new and yield potential applications elsewhere.

This paper is organized as follows.
In Section~\ref{sec:main}, we present the proof of the key trace inequality \eqref{eq:goal}.
In Section~\ref{sec:applictaion}, we demonstrate the applications of \eqref{eq:goal} in various quantum information processing tasks.


\section{Main Result: A Novel Trace Inequality} \label{sec:main}

\begin{theo} \label{theo:sharp_one-shot}Let $A$ and $B$ be positive semi-definite trace-class operators on a infinite-dimensional separable Hilbert space. Suppose the support of $A$ is contained in support of $B$ and  $\Tr\left[ A \left( \log(A+B) - \log B \right) \right] < \infty$. Then
\begin{align}
	\Tr\left[ A \left( \log(A+B) - \log B \right) \right]
    &\leq
	c_s \int_0^\infty\Tr\left[ \left( A (B + t \I )^{-1} \right)^{1+s}  \right]\d t
    \\
    &\leq\frac{c_s}{s}\Tr\left[ \left( B^{-\frac{s}{2(1+s)}} A B^{-\frac{s}{2(1+s)}} \right)^{1+s} \right], \quad \forall\, s\in (0,1],
    \label{eq:sharp_one-shot}
\end{align}
where $c_s=s^s(1-s)^{1-s} \leq 1$ for all $s\in[0,1]$. 
\end{theo}
\begin{proof} \label{proof:sharp_one-shot}
We first prove our claim for finite-dimensional Hilbert spaces, and then employ the finite-rank approximations to extend our results to infinite dimensions \cite[\S III.C]{Mos23}.
Note that the support of $A$ must be contained in that of $B$; otherwise, the finiteness hypothesis of $\Tr\left[ A \left( \log(A+B) - \log B \right) \right]$ would be violated.

After confining the space to the positive support of $B$, we may suppose $B>0$.
By the recently established operator layer cake theorem given in Theorem~\ref{theo:Dlog_formula} below with $X\leftarrow A+B$, $Y\leftarrow B$ and the fundamental theorem of calculus, we have
\begin{align*}
	\log (A+B) - \log B
	&= \int_{0}^1 \mathrm{D}\log\left[ B + \beta A \right] (A) \, \d \beta
	\\
&= \int_0^1\int_{0}^\infty  \left\{ u (B+\beta A)< A   \right\}  \, \d u\,\d\beta 
    \\
&\overset{\textnormal{(a)}}{=} \int_0^1\int_{0}^{\nicefrac{1}{\beta}}  \left\{ u (B+\beta A)< A   \right\}  \, \d u\,\d\beta 
    \\
&\overset{\textnormal{(b)}}{=} \int_0^1\int_{0}^\infty  \left\{ A > \gamma B  \right\} \frac{1}{(1+\beta\gamma)^2} \, \d \gamma\,\d\beta 
	\\
&\overset{\textnormal{(c)}}{=} \int_0^\infty \left\{ A > \gamma B \right\} \frac{1}{\gamma + 1} \, \d \gamma.
\end{align*}
Here, in (a), we have, for $u\ge \nicefrac{1}{\beta}$,
\[ u (B+\beta A)\ge  \frac{1}{\beta} (B+\beta A) > A\pl. \]
In (b), we used the change of variable $\gamma=\frac{u}{1-u\beta}=\frac{1}{\beta}(\frac{1}{1-\mu\beta}-1), u\in [0,\nicefrac{1}{\beta}]$. 
In (c), we calculate: 
\begin{align*}&\int_{0}^1\frac{1}{(1+\beta\gamma)^2 }\,\d\beta =-\frac{1}{\gamma}\frac{1}{(1+\beta\gamma)}\Big\vert_{\beta=0}^{\beta=1}=-\frac{1}{\gamma(1+\gamma)}+\frac{1}{\gamma}= \frac{1}{1+\gamma} .
\end{align*}

By Young's inequality,
\[
\gamma+1=(1-s)\cdot\frac{\gamma}{1-s}+s\cdot\frac{1}{s}\geq \left(\frac{\gamma}{1-s}\right)^{1-s}\left(\frac{1}{s}\right)^s,
\]
which translates to
\[
\frac{1}{\gamma+1}\leq c_s\gamma^{s-1}.
\]

As a result,
\begin{align}
    \int_0^{\infty} \left\{ A > \gamma B \right\} \frac{1}{\gamma+1} \, \d \gamma 
    &\leq c_s\int_0^{\infty} \left\{ A > \gamma B \right\} \gamma^{s-1} \, \d \gamma\\
    &=c_s\int_0^{\infty} \frac{1}{B+t\I}\left(A\frac{1}{B+t\I}\right)^{s}\, \d t.
\end{align}
The equality comes from the change of variables (see Theorem~\ref{theo:change-of-variable} below).

Hence,
\begin{align*}
    \Tr\left[ A \left( \log(A+B) - \log B \right) \right]
    &\leq c_s\Tr\left[ A \int_0^{\infty} \frac{1}{B+t\I}\left(A\frac{1}{B+t\I}\right)^{s}\, \d t \right]\\
    &=c_s \int_0^\infty\Tr\left[ \left( A (B + t \I )^{-1} \right)^{1+s}  \right]\d t\\
    &=c_s \int_0^\infty\Tr\left[ \left( (B + t \I )^{-\frac{1}{2}} A (B + t \I )^{-\frac{1}{2}} \right)^{1+s}  \right]\d t\\
    &\leq \frac{c_s}{s}\Tr\left[ \left(B^{-\frac{s}{2(1+s)}}AB^{-\frac{s}{2(1+s)}} \right)^{1+s} \right].
\end{align*}
The last inequality comes from the Araki--Lieb--Thirring inequality (see e.g., \cite[Proposition~3.10]{beigi2025some} but for $\alpha = 1+s>1$).

We now extend \eqref{eq:sharp_one-shot} to infinite-dimensional Hilbert spaces.
Let $(P_n)_{n\in\mathds{N}}$ be a sequence of finite rank spectral projections of $B$ such that $P_{n-1} \leq P_n$ and $P_n \nearrow \I$ in the strong operator topology.
Denote by $A_n = P_n A P_n$, $B_n = P_n B P_n$, and define
$D(A \Vert B) := \Tr\left[ A (\log A - \log B )] +\Tr[ B - A\right]$ as the Lindblad extension of relative entropy to positive semi-definite operators.
We start with the left-hand side:
\begin{align}
&\Tr\left[ A \left( \log (A+B) - \log B \right)\right]
\notag
\\
&= \Tr\left[ (A+B) \left( \log (A+B) - \log B \right) \right] - \Tr\left[ B \left( \log (A+B) - B \log B \right) \right]
\\
&= D(A+B\Vert B) + D(B\Vert A+B)
\\
&= \lim_{n\to\infty} \left\{ D\left( A_n+B_n \Vert B_n \right) 
+ D( B_n \Vert A_n+B_n ) \right\},
\end{align}
where we used Lemma~\ref{lemm:finit-rank} for approximating $D(\cdot\Vert\cdot)$ in the last line. Note that here the first equality is well-defined as we never run into ``$\infty-\infty$", because $D(B\Vert A+B)$ is always finite as $B\le A+B$.

On the other hand, for any integer $n\in\mathds{N}$,
we have shown 
\begin{align*}
    D\left( A_n+B_n \Vert B_n \right) 
    + D\left( B_n \Vert A_n+B_n \right) 
    &=
    \Tr\left[ A_n \left( \log(A_n+B_n) - \log B_n \right) \right]
    \\
    &\leq 
    c_s \Tr\left[ \left( (B_n + t \I )^{-\frac{1}{2}} A_n (B_n + t \I )^{-\frac{1}{2}} \right)^{1+s}  \right]
    \\
    &\leq \frac{c_s}{s}\Tr\left[ \left(B_n^{-\frac{s}{2(1+s)}} A_n B_n ^{-\frac{s}{2(1+s)}} \right)^{1+s} \right]
    \\
    &=: \frac{c_s}{s} \widetilde{Q}_{1+s}(A_n\Vert B_n), \quad \forall\, s\in (0,1].
\end{align*}
By applying Lemma~\ref{lemm:finit-rank-integral} for approximating the intermediate term $\Tr[ ( (B + t \I )^{-\nicefrac{1}{2}} A (B + t \I )^{-\nicefrac{1}{2}} )^{1+s}  ]$ and 
Lemma~\ref{lemm:finit-rank} again for approximating $\widetilde{Q}_{1+s}(\cdot\Vert\cdot)$, 
we conclude the proof.
\end{proof}

\begin{theo}[Operator layer cake {\cite[Theorem B.1]{preparation}}] \label{theo:Dlog_formula}
For any positive definite operator $X$ and any positive semi-definite operator $Y$ on a finite-dimensional Hilbert space,
the following representation holds:
\begin{align} \label{eq:Dlog_formula}
\mathrm{D} \log[X](Y) = \int_{0}^{\infty} \{ uX < Y \} \d u,
\end{align}
where $ \mathrm{D} \log[X](Y) $ is the directional derivative of the operator logarithm at $X$ with direction $Y$,
and $ \{ uX < Y \} \equiv \{ Y - uX > 0 \} $ denotes the projection onto the positive part of $Y-uX$.
\end{theo}

\begin{theo}[Operator change of variables {\cite[Theorem~C.1]{preparation}}]  \label{theo:change-of-variable}
Let $A$ and $B$ be finite-dimensional positive semi-definite operators satisfying $r := \|A B^{-1} \|_{\infty} < \infty$.
Then, for any Lebesgue-integrable function $h$ on $[0,r]$,
\begin{align} \label{eq:change-of-variable}
	\int_0^r \proj{ A > \gamma B } h(\gamma) \, \d \gamma
    &= \int_0^\infty \frac{1}{B+t\I} A \frac{1}{B+t\I} h\left( A \frac{1}{B+t\I} \right) \d t.
\end{align}
\end{theo}

\begin{remark}
The operators \( A(B + t\I)^{-1} \) and \( (B + t\I)^{-1}A \) are diagonalizable and have the same spectrum as \( (B + t\I)^{-1/2} A (B + t\I)^{-1/2} \) for all $t\geq 0$, since they are all similar. This implies that
\[
\frac{1}{B+t\I} h\left(A\frac{1}{B+t\I}\right)
= \frac{1}{\sqrt{B+t\I}} h\left(\frac{1}{\sqrt{B+t\I}} A \frac{1}{\sqrt{B+t\I}}\right)\frac{1}{\sqrt{B+t\I}}
= h\left(\frac{1}{B+t\I} A\right)\frac{1}{B+t\I}.
\]
Hence, each integrand in the right-hand side of \eqref{eq:change-of-variable} is self-adjoint, i.e.,
\begin{align}
    \frac{1}{\sqrt{B+t\I}} \underbrace{\frac{1}{\sqrt{B+t\I}} A  \frac{1}{\sqrt{B+t\I}} h\left(\frac{1}{\sqrt{B+t\I}} A \frac{1}{\sqrt{B+t\I}}\right)}_{ \text{self-adjoint} } \frac{1}{\sqrt{B+t\I}}.
\end{align}
\end{remark}

\section{Applications} \label{sec:applictaion}

In this section, we demonstrate how the established trace inequality in Theorem~\ref{theo:sharp_one-shot} sharpens the existing one-shot achievability bounds in various quantum information-theoretic tasks including classical-quantum soft covering (Section~\ref{sec:covering}),
privacy amplification against quantum side information (Section~\ref{sec:PA}),
convex splitting (Section~\ref{sec:splitting}),
quantum information decoupling (Section~\ref{sec:decoupling}),
as well as quantum channel simulation (Section~\ref{sec:simulation}).
We will express the error estimates in terms of the integral \Renyi divergence \cite{hirche2023quantum, beigi2025some}
\begin{align}
    D_{\alpha}(\rho\Vert\sigma)
    \coloneq \frac{1}{\alpha-1}\log (\alpha-1) \int_0^\infty \Tr\left[ \left( B + t\I\right)^{-\nicefrac{1}{2}} A \left( B + t\I\right)^{-\nicefrac{1}{2}}  \right]\d t
\end{align}
or the sandwiched \Renyi divergence \cite{MDS+13, WWY14}:
\begin{align}
    \widetilde{D}_{\alpha} (\rho\Vert\sigma) 
    \coloneq \frac{1}{\alpha-1}\log \Tr\left[ \left( \sigma^{-\frac{1}{2\alpha}} \rho  \sigma^{-\frac{1}{2\alpha}} \right)^{\alpha} \right], \quad \alpha > 1
\end{align}
where $\rho$ and $\sigma$ are positive semi-definite trace-class operators with $\Tr[\rho] = 1$.
There, the error criterion is either under the purified distance 
\begin{align}
    P(\rho, \sigma) := \sqrt{  1 - \e^{- \widetilde{D}_{1/2}(\rho\Vert\sigma)}  },
\end{align}
or the quantum relative entropy \cite{Ume56}:
\begin{align}
D(\rho\Vert\sigma) = \lim_{\alpha \searrow 1 } \widetilde{D}_{\alpha} (\rho\Vert\sigma) 
= \Tr\left[ \rho ( \log \rho - \log \sigma ) \right].
\end{align}

\subsection{Soft Covering} \label{sec:covering}


\begin{defn}[Classical-quantum soft covering with non-uniform randomness] \label{defn:covering}
	Let $\rho_{\X\B} = \sum_{x\in\X} p_{\X}(x) |x\rangle \langle x|_{\X} \otimes \rho_{\B}^x$ be a classical-quantum state, where $p_{\X}$ is a probability distribution on a finite alphabet $\X$, and each $\rho_{\B}^x$ is a density operator (i.e.~a positive semi-definite operator with unit trace), and the marginal state on system $\B$ is $\rho_{\B} = \sum_{x\in\X} p_{\X} (x) \rho_{\B}^x$.
    Let $p_{\mathsf{M}}$ be a probability distribution on an alphabet $\mathsf{M}$.
	\begin{enumerate}[1.]
		\item Alice has classical registers $\mathsf{M}$ and $\X$.
		
		\item Alice samples from the set $\mathsf{M}$ according distribution $p_{\mathsf{M}}$.
		
		\item For each sample $m\in\mathsf{M}$, Alice encodes it to a codeword $x(m)$ in $\X$.
		
		\item Alice queries the classical-quantum channel $x\mapsto \rho_{\B}^x$ with the codeword $x(m)$.
	\end{enumerate}
	
	An $(M, \eps)$-resolvability code is a codebook $\{x(m)\}_{m\in\mathsf{M}}$ satisfying $|\mathsf{M}| = M$ such that the codebook-induced state $\mathds{E}_{m \sim p_{\mathsf{M}}}[\rho_{\B}^{x(m)}]$ is at least $\varepsilon$-close to the target state $\rho_{\B}$ in terms of relative entropy, i.e.~
	\begin{align} \notag
		 D\left(\mathds{E}_{m \sim p_{\mathsf{M}}} [\rho_{\B}^{x(m)}] \Vert \rho_{\B} \right)
		\leq \varepsilon.
	\end{align}
\end{defn}

We adopt the random coding as follows.
For each $m\in\mathsf{M}$, Alice chooses the codeword $x(m)$ according to the input distribution $p_{\X}$ pairwise independently, i.e., the random codeword $x(m)$ is independent of $x(\bar{m})$ for $m\neq \bar{m}$.
Channel resolvability via random coding is called \emph{soft covering}; see \cite{Hay15, CG22, SGC22b, HCG24}.

\begin{prop} \label{prop:covering}
For any classical-quantum state $\rho_{\X\B} = \sum_{x\in\X} p_{\X}(x) |x\rangle \langle x|_{\X} \otimes \rho_{\B}^x$ and distribution $p_{\mathsf{M}}$ given in Definition~\ref{defn:covering}, the random coding error satisfies
\begin{align} \label{eq:covering}
\mathds{E}_{x(m)\sim p_{\X}} D\left(\mathds{E}_{m \sim p_{\mathsf{M}}} [\rho_{\B}^{x(m)}] \Vert \rho_{\B} \right)
\leq \frac{c_{\alpha-1}}{\alpha-1} \e^{ - (\alpha-1) \left[ H_{\alpha}(\mathsf{M})_{p} - {D}_{\alpha}\left( \rho_{\X\B} \Vert \rho_{\X} \otimes \rho_{\B} \right) \right]  }, \quad \forall\,\alpha \in (1,2],
\end{align}
where $H_{\alpha}(\mathsf{M})_{p} := \frac{1}{1-\alpha}\log \sum_m p_{\mathsf{M}}(m)^{\alpha}$ is the R\'enyi entropy.
\end{prop}

Proposition~\ref{prop:covering} improves on \cite[Lemma 4]{Hay15} by a factor $c_{\alpha-1} \in [1/2,1]$ for $\alpha \in (1,2]$ and by removing the dimension-dependent factor $|\text{spec}(\mathcal{H}_{\B})|^{\alpha-1}$. This in turn improves mutual information leakage to quantum eavesdroppers via a classical-quantum wiretap channel by the same fashion; c.f.~\cite[(65)]{Hay15}.

If uniform randomness is available at Alice, i.e.,
$p_{\mathsf{M}}$ is a uniform distribution,  Definition~\ref{defn:covering} reduces to the conventional {classical-quantum channel resolvability} via uniform randomness, and the right-hand side of \eqref{eq:covering} becomes $\frac{c_{\alpha-1}}{\alpha-1} \e^{ - (\alpha-1) \left[ \log |\mathsf{M}| - \widetilde{D}_{\alpha}(\rho_{\X\B}\Vert \rho_{\X} \otimes \rho_{\B}) \right]  }$.
In the $n$-fold independent and identical setting where $\rho_{\X\B}\leftarrow \rho_{\X\B}^{\otimes n}$ and $|\mathsf{M}| = \exp(nR)$ with $R > I(\X:\B)_{\rho} = D(\rho_{\X\B}\Vert \rho_{\X}\otimes \rho_{\B})$, we remark that the (regularized) error exponent obtained in Proposition~\ref{prop:covering}, i.e.,
\begin{align}
\sup_{\alpha \in (1,2]} (\alpha-1) \left[ R - \widetilde{D}_{\alpha}(\rho_{\X\B}\Vert \rho_{\X} \otimes \rho_{\B}) \right]
\end{align}
is tight for the commuting case \cite[Theorem 3]{TPM17}.

\begin{proof}[Proof of Proposition~\ref{prop:covering}]
The first part of the proof essentially follows \cite[Lemma 4]{Hay15}.
Given each $m\in\mathsf{M}$ and the corresponding realization of a codeword $x(m) \in \X $, we first calculate the conditional expectation:
\begin{DispWithArrows}[displaystyle]
&\mathds{E}_{x(\bar{m})\mid x(m)} \left[ \log \sum_{\bar{m}\in \mathsf{M}} p_{\mathsf{M}}(\bar{m}) \rho_{\B}^{x(\bar{m})}\right]
\notag
\\
&= \mathds{E}_{x(\bar{m})\mid x(m)} \left[ \log \left( p_{\mathsf{M}}(m) \rho_{\B}^{x(m)} + \sum_{\bar{m}\neq m} p_{\mathsf{M}}(\bar{m}) \rho_{\B}^{x(\bar{m})} \right)\right]
\Arrow{\textnormal{$\log$ is operator concave}}
\notag
\\
&\leq \log \left( p_{\mathsf{M}}(m) \rho_{\B}^{x(m)} + \mathds{E}_{x(\bar{m})\mid x(m)} \left[ \sum_{\bar{m}\neq m} p_{\mathsf{M}}(\bar{m}) \rho_{\B}^{x(\bar{m})} \right]\right)
\Arrow{\textnormal{pairwise independence}}
\notag
\\
&= \log \left( p_{\mathsf{M}}(m) \rho_{\B}^{x(m)} + \sum_{\bar{m}\neq m} p_{\mathsf{M}}(\bar{m})  \rho_{\B} \right)
\Arrow{\textnormal{$\sum_{\bar{m}\neq m} p_{\mathsf{M}}(\bar{m})\leq 1$ \&} \\  \textnormal{$\log$ is operator monotone}}
\notag
\\
&\leq \log \left( p_{\mathsf{M}}(m) \rho_{\B}^{x(m)} + \rho_{\B} \right).
\notag
\end{DispWithArrows}

Using the above operator inequality, we have 
\begin{DispWithArrows}[displaystyle]
\mathds{E}_{x(m)\sim p_{\X}} D\left(\mathds{E}_{m \sim p_{\mathsf{M}}} [\rho_{\B}^{x(m)}] \Vert \rho_{\B} \right)
&= \sum_{m\in\mathsf{M}} \mathds{E}_{x(m)} \Tr \left[  
 p_{\mathsf{M}}(m)  \rho_{\B}^{x(m)} \mathds{E}_{x(\bar{m})\mid x(m)} \left( \log \sum_{\bar{m}\in\mathsf{M}} p_{\mathsf{M}}(\bar{m}) \rho_{\B}^{x(\bar{m})} - \log \rho_{\B} \right) \right]
 \notag
\\
&\leq \sum_{m\in\mathsf{M}} \mathds{E}_{x(m)} \Tr  \left[  p_{\mathsf{M}}(m)  \rho_{\B}^{x(m)} \left(  \log \left( p_{\mathsf{M}}(m) \rho_{\B}^{x(m)} + \rho_{\B} \right) - \log \rho_{\B} \right) \right]
\notag
\\
&\leq \frac{c_{\alpha-1}}{\alpha-1} \sum_{m\in\mathsf{M}} \mathds{E}_{x(m)}   p_{\mathsf{M}}(m)^{\alpha} \e^{ (\alpha-1) \widetilde{D}_{\alpha} \left( \rho_{\B}^{x(m)} \Vert \rho_{\B} \right) }
\notag
\\
&= \frac{c_{\alpha-1}}{\alpha-1} \e^{ - (\alpha-1) \left[ H_{\alpha}(\mathsf{M})_{p} - \widetilde{D}_{\alpha}\left( \rho_{\X\B} \Vert \rho_{\X} \otimes \rho_{\B} \right) \right]  }, \quad \forall\,\alpha \in (1,2],
\notag
\end{DispWithArrows}
where the second inequality follows from Theorem~\ref{theo:sharp_one-shot} with $A\leftarrow p_{\mathsf{M}}(m) \rho_{\B}^{x(m)} $ and $B\leftarrow \rho_{\B}$.
\end{proof}

\subsection{Privacy Amplification} \label{sec:PA}

\begin{defn} \label{defn:PA}
Let $\rho_{\X\E} = \sum_{x\in\X} p_{\X}(x) |x\rangle \langle x|_{\X} \otimes \rho_{\E}^x$ be a classical-quantum state.
\begin{enumerate}[1.]
    \item 
    Alice has a classical register $\X$ and the eavesdropper has a quantum register $\E$.
    Initially, they share the state $\rho_{\X\E}$.

    \item
    Alice applies a linear operation $\mathscr{R}^h( \rho_{\X\E} )$ on her system according to a hash function $h:\X\to \Z$:
    \begin{align} 
    \mathscr{R}^h( \rho_{\X\E} )
    := \sum_{x\in\X} p_{\X}(x) |h(x)\rangle \langle h(x)|_{\Z} \otimes \rho_{\E}^x
    = \sum_{z\in\Z} |z\rangle \langle z |_{\Z} \otimes \sum_{x: h(x) = z} p_{\X}(x) \rho_{\E}^x.
    \label{eq:R^h}
    \end{align}
\end{enumerate}
The aim of Alice is for her resulting state $\mathscr{R}^h( \rho_{\X\E} )$ to be independent to the quantum system $\E$ and close to uniform randomness in relative entropy:
\begin{align} \label{eq:error_PA}
D\left( \mathscr{R}^h( \rho_{\X\E} ) \Vert \nicefrac{\I}{|\mathsf{Z}|} \otimes \rho_{\E} \right)
\end{align}
with as larger $|\Z|$ as possible.
\end{defn}

As noted in \cite[Equation (9)]{Hay15_PA}, the security criterion given in \eqref{eq:error_PA} ensures the mutual information between the systems $\Z$ and $\E$ to be controlled, i.e.
\begin{align}
D\left( \mathscr{R}^h( \rho_{\X\E} ) \Vert \nicefrac{\I}{|\mathsf{Z}|} \otimes \rho_{\E} \right)
&= I(\Z : \E)_{ \mathscr{R}^h( \rho_{\X\E} ) } + D\left(\mathscr{R}^h( \rho_{\X} )\Vert \nicefrac{\I}{|\mathsf{Z}|} \right).
\end{align}

We adopt a \emph{$2$-universal random hash function} $h:\X\to \Z$ satisfying for all $x,\bar{x} \in \X$ with $x\neq \bar{x}$,
\begin{align} \label{eq:2-universal}
\Pr_{h}\left\{  h(x) = h(\bar{x}) \right\} \leq \frac{1}{|\Z|}.
\end{align}

\begin{prop} \label{prop:PA}
Following Definition~\ref{defn:PA} and using a $2$-universal random hash function, we have
\begin{align}
\mathds{E}_{h} D\left( \mathscr{R}^h( \rho_{\X\E} )  \Vert \nicefrac{\I}{|\mathsf{Z}|} \otimes \rho_{\E} \right)
&\leq
\frac{c_{\alpha-1}}{\alpha-1}  \e^{ - (\alpha-1) \left[ -\log |\Z| - {D}_{\alpha} \left( \rho_{\X\E} \Vert \I_{\X}\otimes \rho_{\E} \right)  \right] }, \quad \forall\, \alpha \in (1,2]
\end{align}
\end{prop}

Proposition~\ref{prop:PA} improves on \cite[Theorem 1]{Hay15_PA} by a factor $c_{\alpha-1} \in [\nicefrac{1}{2},1]$ for $\alpha \in (1,2]$ and by removing the dimension-dependent factor $|\text{spec}(\mathcal{H}_{\E})|^{\alpha-1}$ at the eavesdropper.

In the $n$-fold independent and identical setting where $\rho_{\X\E}\leftarrow \rho_{\X\E}^{\otimes n}$ and $|\Z| = \exp(nR)$ with $R < H(\X\mid\E)_{\rho} := - D(\rho_{\X\E}\Vert \I_{\X}\otimes \rho_{\E})$, we remark that the (regularized) error exponent obtained in Proposition~\ref{prop:PA}, i.e.,
\begin{align}
\sup_{\alpha \in (1,2]} (\alpha-1) \left[ -\widetilde{D}_{\alpha}(\rho_{\X\E}\Vert \I_{\X} \otimes \rho_{\E}) - R \right]
\end{align}
is tight for $\left.\frac{\d}{\d s} - s \widetilde{D}_{1+s}(\rho_{\X\E}\Vert\I_{\X}\otimes \rho_{\E})\right|_{s=1} \leq R < H(\X\mid\E)_{\rho}$ \cite[Theorem 1]{KL21} (see also \cite[Theorem 1]{HT17} for the classical case).

\begin{proof}[Proof of Proposition~\ref{prop:PA}]
    The proof follows closely from that of Proposition~\ref{prop:covering}, which is also inspired by that of \cite[Theorem 1]{Hay15_PA}.
    The key difference is that we will employ Theorem~\ref{theo:sharp_one-shot} in the derivations.

    For each fixed $x\in\X$, the operator concavity of logarithm implies
    \begin{DispWithArrows}[displaystyle]
    \mathds{E}_h \log \left( \sum_{\bar{x}\colon h(\bar{x}) = h(x) } p_{\X}(\bar{x}) \rho_{\E}^{\bar{x}} \right)
    &\leq \log \left( \mathds{E}_h \left[ \sum_{\bar{x}\colon h(\bar{x}) = h(x) } p_{\X}(\bar{x}) \rho_{\E}^{\bar{x}} \right] \right)
    \notag
    \\
    &=\log \left( \mathds{E}_h \left[ p_{\X}({x}) \rho_{\E}^{{x}}  + \sum_{\substack{\bar{x}\colon h(\bar{x}) = h(x) \\ \bar{x}\neq x} } p_{\X}(\bar{x}) \rho_{\E}^{\bar{x}} \right]\right)
    \notag
    \\
    &\overset{\textnormal{(a)}}{\leq} \log \left(  p_{\X}({x}) \rho_{\E}^{{x}}  + \frac{1}{|\Z|} \sum_{\bar{x} \colon \bar{x}\neq x } p_{\X}(\bar{x}) \rho_{\E}^{\bar{x}} \right)
    \notag
    \\
    &\overset{\textnormal{(b)}}{\leq}  \log \left(  p_{\X}({x}) \rho_{\E}^{{x}}  + \frac{1}{|\Z|} \rho_{\E} \right), 
    \label{eq:PA_1}
    \end{DispWithArrows}
    where (a) is by the definition of $2$-universal hash functions in \eqref{eq:2-universal} and the operator monotonicity of logarithm, 
    and (b) is because $\sum_{\bar{x}\neq x} p_{\X}(\bar{x}) \rho_{\E}^{\bar{x}} \leq \sum_{\bar{x} \in \X} p_{\X}(\bar{x}) \rho_{\E}^{\bar{x}} = \rho_{\E}$ and again logarithm is operator monotone.

    Then, using \eqref{eq:R^h}, we calculate
    \begin{DispWithArrows}[displaystyle]
    \mathds{E}_{h} D\left( \mathscr{R}^h( \rho_{\X\E} )  \Big\Vert \frac{\I}{|\mathsf{Z}|} \otimes \rho_{\E} \right)
    &{=} \mathds{E}_h D \left( \sum_{z\in\Z} |z\rangle \langle z|_{\X} \otimes \sum_{x\colon h(x) = z} p_{\X}(x)\rho_{\E}^x \,\bigg\Vert\, \frac{\I}{|\Z|}\otimes \rho_{\E} \right)
    \\
    &\overset{\textnormal{(a)}}{=} \mathds{E}_h \sum_{z\in\Z} \Tr\left[ \sum_{x\colon h(x) = z} p_{\X}(x)\rho_{\E}^x \left( \log\left( \sum_{ \bar{x} \colon h(\bar{x}) = z} p_{\X}(\bar{x})\rho_{\E}^{\bar{x}} \right) - \log \frac{\rho_{\E}}{|\Z|} \right) \right]
    \\
    &= \mathds{E}_h \sum_{x\in\X} \Tr\left[  p_{\X}(x)\rho_{\E}^x \left( \log\left( \sum_{ \bar{x} \colon h(\bar{x}) = h(x) } p_{\X}(\bar{x})\rho_{\E}^{\bar{x}} \right) - \log \frac{\rho_{\E}}{|\Z|} \right) \right]
    \\
    &\overset{\textnormal{(b)}}{\leq} \sum_{x\in\X} \Tr\left[  p_{\X}(x)\rho_{\E}^x \left( \log \left(  p_{\X}({x}) \rho_{\E}^{{x}}  + \frac{\I}{|\Z|} \rho_{\E} \right) - \log \frac{\rho_{\E}}{|\Z|} \right) \right]
    \\
    &\overset{\textnormal{(c)}}{\leq} \frac{c_{\alpha-1}}{\alpha-1}  \sum_{x\in\X} p_{\X}(x)^{\alpha} \e^{ (\alpha-1) \left[ \widetilde{D}_{\alpha}\left( \rho_{\E}^x \Vert \rho_{\E} \right) + \log |\Z| \right]  }
    \\
    &= \frac{c_{\alpha-1}}{\alpha-1}  \sum_{x\in\X} p_{\X}(x)^{\alpha} \e^{ (\alpha-1) \left[ \widetilde{D}_{\alpha}\left( \rho_{\E}^x \Vert \rho_{\E} \right) + \log |\Z| \right]  }, 
    \\ &=\frac{c_{\alpha-1}}{\alpha-1}  \e^{ - (\alpha-1) \left[ -\log |\Z| - \widetilde{D}_{\alpha} \left( \rho_{\X\E} \Vert \I_{\X}\otimes \rho_{\E} \right)  \right] }
    \quad \forall\, \alpha \in (1,2],
    \end{DispWithArrows}
    where 
    (a) follows from the direct-sum structure of $D(\cdot\Vert \cdot)$,
    (b) follows from \eqref{eq:PA_1},
    and (c) follows from Theorem~\ref{theo:sharp_one-shot} with $A\leftarrow p_{\X}(x) \rho_{\E}^x$ and $B\leftarrow {\rho_{\E}}/{|\Z|}$.
\end{proof}

\subsection{Convex Splitting} \label{sec:splitting}

\begin{defn}[Convex splitting with non-uniform randomness] \label{defn:splitting}
Let $\rho_{\A\B}$ and $\tau_{\A}$ be quantum states satisfying $\textrm{supp}(\rho_{\A}) \subseteq \textrm{supp}(\tau_{\A})$, let $\mathsf{M} = \{1,2,\ldots, M\} =: [M]$ be a finite set, and let $p_{\mathsf{M}}$ be a probability distribution on $\mathsf{M}$.
\begin{enumerate}[1.]
    \item 
    Alice has quantum registers $\A_1, \A_2, \ldots, \A_M$, where $\A_m \simeq \A$ all initialized with state $\tau_{\A_m}$ and has a quantum register $\A$, and Bob has a quantum register $\B$.
    The initial state on system $\A\B$ is $\rho_{\A\B}$.

    \item
    Alice randomly embeds her state on $\A$ to $\A_m$ with probability $p_{\mathsf{M}}(m)$.
\end{enumerate}
    The aim of Alice is for the mixture
    \begin{align}
	\omega_{\A_1\ldots \A_M \B} = \sum_{m=1}^M p_{\mathsf{M}} (m) \cdot \rho_{\A_m \B} \otimes \tau_{\A}^{\otimes [M]\backslash \{m\}}
    \end{align}
    to be close to the product state $\otimes_{m=1}^M \tau_{\A_m} \otimes \rho_{\B}$ in relative entropy:
    \begin{align}
	\begin{split} \label{eq:D}
		&\Delta^{D}_{M}\left( \rho_{AB} \,\Vert\, \tau_{\A} \right) :=
		D\left( \omega_{\A_1\ldots \A_M \B} \, \Vert\, \tau_{\A}^{\otimes M} \otimes \rho_B \right) 
        \notag
        \\
		&= \sum_{m\in\mathsf{M}} \Tr\left[  p_{\mathsf{M}}(m) \rho_{\A_m \B} \otimes \tau_{\A}^{\otimes[M]\backslash \{m\}} \left( \log \left( \sum_{\bar{m}\in\mathsf{M}} p_{\mathsf{M}}(\bar{m}) \rho_{\A_{\bar{m}} \B} \otimes \tau_{\A}^{\otimes[M]\backslash \{\bar{m}\}} \right) - \log \tau_{\A}^{\otimes M} \otimes \rho_B \right) \right]
	\end{split}
    \end{align}
    as least integer $M$ as possible.
\end{defn}

We remark that convex splitting was proposed by Anshu \textit{et al.} \cite{ADJ17}.
A tight analysis under trace distance was studied by parts of the authors \cite{CG22b}.

\begin{prop} \label{prop:splitting}
Let $\rho_{\A\B}$ and $\tau_{\A}$ be quantum states satisfying $\textrm{supp}(\rho_{\A}) \subseteq \textrm{supp}(\tau_{\A})$.
Following Definition~\ref{defn:splitting}, we have
\begin{align}
\Delta^{D}_{M}\left( \rho_{\A\B} \,\Vert\, \tau_{\A} \right)
&\leq
\frac{c_{\alpha-1}}{\alpha-1}  \e^{ - (\alpha-1) \left[ H_{\alpha}(\mathsf{M})_p - {D}_{\alpha} \left( \rho_{\A\B} \Vert \tau_{\A}\otimes \rho_{\B} \right)  \right] }, \quad \forall\, \alpha \in (1,2]
\end{align}
\end{prop}

Proposition~\ref{prop:splitting} improves on 
\cite[Lemma 13]{LY21b} by a factor $c_{\alpha-1} \in [\nicefrac{1}{2},1]$ for $\alpha \in (1,2]$ and by removing the dimension-dependent factor $|\text{spec}(\mathcal{H}_{\A} \otimes \mathcal{H}_{\B})|^{\alpha-1}$.

\begin{proof}
First, for each $m\in [M]$, we define a completely positive trace-preserving map to trace out all the $\A_{\bar{m}}$ systems for $\bar{m}\neq m$ and append it with a state $\tau_{\A_{\bar{m}}}$, i.e.,
\begin{align}
	\mathscr{N}^{(m)} = \Tr_{ \A_{[M]\backslash \{m\} } }\left[\,\cdot\, \right] \otimes \tau_{\A}^{\otimes [M]\backslash \{{m}\}},
\end{align}
such that
\begin{align} \label{eq:splitting_trace_out}
\begin{split}
	\mathscr{N}^{(m)} \left( \rho_{\A_{\bar{m}} \B} \otimes \tau_{\A}^{\otimes [M]\backslash \{\bar{m}\}} \right)
	=
	\begin{dcases}
		\rho_{\A_{{m}} B} \otimes \tau_{\A}^{\otimes[M]\backslash \{{m}\}} & \bar{m} = m, \\
		\tau_{\A}^{\otimes M} \otimes \rho_B & \bar{m} \neq m. \\
	\end{dcases}
    \end{split}
\end{align}
Then, data-processing inequality of Umegaki's relative entropy \cite{Ume54} implies that, for each $m\in[M]$,
\begin{align}
	D\left( \rho_{\A_m \B} \otimes \tau_{\A}^{\otimes[M]\backslash \{m\}} \, \Vert\, \omega_{\A_1\ldots \A_M \B} \right)
	&\geq D\left( \mathscr{N}^{(m)} \left( \rho_{\A_m \B} \otimes \tau_{\A}^{\otimes[M]\backslash \{m\}}  \right) \, \Vert\, \mathscr{N}^{(m)} \left( \omega_{\A_1\ldots \A_M \B}  \right) \right)\\
	&= D\left( \rho_{\A_m \B} \otimes \tau_{\A}^{\otimes[M]\backslash \{m\}}  \, \Vert\, \mathscr{N}^{(m)} \left( \omega_{\A_1\ldots \A_M \B}  \right) \right),
\end{align}
which translates to
\begin{align}
	\Tr\left[ \rho_{\A_m \B} \otimes \tau_{\A}^{\otimes[M]\backslash \{m\}} \log \omega_{\A_1\ldots \A_M \B} \right] \leq
	\Tr\left[ \rho_{\A_m \B} \otimes  \tau_{\A}^{\otimes[M]\backslash \{m\}} \log  \mathscr{N}^{(m)} \left(\omega_{\A_1\ldots \A_M \B} \right) \right]. \label{eq:D1}
\end{align}

By applying \eqref{eq:D1}, we bound the first term in the bracket of \eqref{eq:D} as follows: for each $m\in[M]$,
\begin{align}
	&\Tr\left[ p_{\mathsf{M}}(m) \rho_{\A_m \B} \otimes \tau_{\A}^{\otimes[M]\backslash \{m\}} \log \omega_{\A_1\ldots \A_M \B} \right] \\
	&\leq 	\Tr\left[ p_{\mathsf{M}}(m) \rho_{\A_m \B} \otimes \tau_{\A}^{\otimes[M]\backslash \{m\}} \log   \mathscr{N}^{(m)} \left(\omega_{\A_1\ldots \A_M \B} \right) \right] \\
	&\overset{\textnormal{(a)}}{=} \Tr\left[ p_{\mathsf{M}}(m) \rho_{\A_m \B} \otimes \tau_{\A}^{\otimes[M]\backslash \{m\}} \log \left( p_{\mathsf{M}}(m) \rho_{\A_m \B} \otimes \tau_{\A}^{\otimes[M]\backslash \{m\}} + \sum_{\bar{m}\neq m} p_{\mathsf{M}}(\bar{m}) \tau_{\A}^{\otimes M} \otimes \rho_B \right) \right] \\
	&\overset{\textnormal{(b)}}{\leq} \Tr\left[ p_{\mathsf{M}}(m) \rho_{\A_m \B} \otimes \tau_{\A}^{\otimes[M]\backslash \{m\}} \log \left( p_{\mathsf{M}}(m) \rho_{\A_m \B} \otimes \tau_{\A}^{\otimes[M]\backslash \{m\}} +  \tau_{\A}^{\otimes M} \otimes \rho_B \right) \right], \label{eq:D2}
\end{align}
where we invoked \eqref{eq:splitting_trace_out} in (a) and used the operator monotonicity of logarithm in (b).
Combining \eqref{eq:D} with \eqref{eq:D2}, we have, for all $s\in[0,1]$:
\begin{DispWithArrows}[displaystyle]
	&D\left( \omega_{\A_1\ldots \A_M \B} \, \Vert \tau_{\A}^{\otimes M} \otimes \rho_B \right)
    \notag
    \\ 
    &\leq \sum_{m\in\mathsf{M}} \Tr\left[ p_{\mathsf{M}}(m) \rho_{\A_m \B} \otimes \tau_{\A}^{\otimes[M]\backslash \{m\}} \left( \log \left( p_{\mathsf{M}}(m) \rho_{\A_m \B} \otimes \tau_{\A}^{\otimes[M]\backslash \{m\}} +  \tau_{\A}^{\otimes M} \otimes \rho_B \right) - \log \tau_{\A}^{\otimes M} \otimes \rho_{\B} \right)\right]
    \notag
    \\
    &= \sum_{m\in\mathsf{M}} \Tr\left[ p_{\mathsf{M}}(m)\rho_{\A\B} \left( \log ( p_{\mathsf{M}}(m) \rho_{\A\B} + \tau_{\A}\otimes \rho_{\B}) - \log \tau_{\A}\otimes \rho_{\B} \right) \right] 
    \notag
    \\
    &\leq \frac{c_{\alpha-1}}{\alpha-1} \sum_{m\in\mathsf{M}} \mathds{E}_{x(m)}   p_{\mathsf{M}}(m)^{\alpha} \e^{ (\alpha-1) \widetilde{D}_{\alpha} \left( \rho_{\A\B} \Vert \tau_{\A}\otimes \rho_{\B} \right) }
    \notag
    \\
    &= \frac{c_{\alpha-1}}{\alpha-1}  \e^{ - (\alpha-1) \left[ H_{\alpha}(\mathsf{M})_p - \widetilde{D}_{\alpha} \left( \rho_{\A\B} \Vert \tau_{\A}\otimes \rho_{\B} \right)  \right] }, \quad \forall\, \alpha \in (1,2],
\end{DispWithArrows}
where the first inequality follows from 
\begin{align*}&\Tr\left[\rho_{\A_m \B} \otimes \tau_{\A}^{\otimes[M]\backslash \{m\}} \left( \log \left( p_{\mathsf{M}}(m) \rho_{\A_m \B} \otimes \tau_{\A}^{\otimes[M]\backslash \{m\}} +  \tau_{\A}^{\otimes M} \otimes \rho_B \right) - \log \tau_{\A}^{\otimes M} \otimes \rho_{\B} \right)\right]
\\=& \Tr\left[\rho_{\A_m \B} \otimes \tau_{\A}^{\otimes[M]\backslash \{m\}} \left( \log \left( (p_{\mathsf{M}}(m) \rho_{\A_m \B}+\tau_{\A}^{\otimes M}) \otimes \tau_{\A}^{\otimes[M]\backslash \{m\}}\right) - \log \tau_{\A}\otimes \rho_{\B} \otimes \tau_{\A}^{\otimes[M]\backslash \{m\}}\right)\right]
\\=& \Tr\left[\rho_{\A_m \B} \otimes \tau_{\A}^{\otimes[M]\backslash \{m\}} \left( \log (p_{\mathsf{M}}(m) \rho_{\A_m \B}+\tau_{\A}^{\otimes M})  - \log \tau_{\A}\otimes \rho_{\B} \right)\right]
\\=& \Tr\left[\rho_{\A_m \B}  \left( \log (p_{\mathsf{M}}(m) \rho_{\A_m \B}+\tau_{\A}^{\otimes M})  - \log \tau_{\A}\otimes \rho_{\B} \right)\right]\ 
\end{align*}
and the second inequality follows from Theorem~\ref{theo:sharp_one-shot} with $A\leftarrow p_{\mathsf{M}}(m) \rho_{\A\B} $ and $B\leftarrow \tau_{\A}\otimes \rho_{\B}$.
\end{proof}

If $p_{\mathsf{M}}$ is uniform, we then have the following achievable (regularized) error exponent for any $n$-fold product expansion: for all $R > D(\rho_{\A\B}\Vert \tau_{\A}\otimes \rho_{\B})$,
\begin{align}
	\lim_{n\to \infty} -\frac1n \log \Delta^{D}_{\mathrm{e}^{nR}}\left( \rho_{AB}^{\otimes n} \,\Vert\, \tau_{\A}^{\otimes n} \right)
	\geq
	\sup_{\alpha\in (1,2]} (\alpha-1) \left( R - \widetilde{D}_\alpha(\rho_{AB}\,\Vert\, \tau_{\A}\otimes \rho_B)  \right).
\end{align}

\subsubsection{Quantum State Redistribution} \label{sec:QSR}


\begin{defn}[Quantum State Redistribution] \label{defn:QSR}

	Let $\rho_{\R \A \C \B}$ be a pure state.
	\begin{enumerate}[1.]
		\item 
        Alice has quantum registers $\A$ and $\C$ at sender,
        Bob has a quantum register $\B$ at receiver, and $\R$ is an inaccessible reference system.
        The initial state of the protocol is $\rho_{\R \A \C \B}$.
		
		\item A resource of free entanglement, say $|\tau\rangle_{\bar{\A} \bar{\B}}$, is shared between the sender (holding register $\bar{\A}$) and the receiver (holding register $\bar{\B}$), and noiseless one-way classical communication from the sender to receiver is available.
		
		\item Alice applies a local operation on her system and the shared entanglement to obtain $\log M$ nats of classical messages.
		
		\item The sender sends the above message to the receiver via one-way noiseless classical communication.
		
		\item 
		Upon receiving the messages, the receiver applies a local operation on his shared entanglement to obtain an overall resulting state $\widehat{\rho}_{\R \A \B \C}$ and now the quantum register $\C$ is held by Bob at the receiver.
\end{enumerate}
    
An $(M, \eps)$ Quantum State Redistribution protocol for $\rho_{\R \A \C}$ with entanglement $|\tau\rangle_{\bar{\A}' \bar{\B}'}$ satisfies
		\begin{align}
			P\left( \widehat{\rho}_{\R \A \B \C} , \rho_{\R \A \B \C} \right) \leq \eps.
		\end{align}
\end{defn}

\begin{prop} \label{prop:QSR}
For any pure state $\rho_{\R\A\C\B} = |\rho\rangle\langle \rho|_{\R\A\C\B}$, there exists an $(M, \eps)$ Quantum State Redistribution protocol for $\rho_{\R\A\C\B}$ with entanglement $|\tau\rangle_{\bar{\A} \bar{\B}}^{\otimes M}$ (where $\bar{\A} \cong \bar{\B} \cong \C $) satisfying
	\begin{align} \label{eq:error_QSS}
	\eps &\leq
	\sqrt{ \frac{c_{\alpha-1}}{\alpha-1} } \cdot
    \e^{- \frac{(\alpha-1)}{2}\left[ \log M - {D}_{\alpha}(\rho_{\C\R\B} \Vert \tau_{\C} \otimes \rho_{\R\B} ) \right]  }, \quad \forall\, \alpha \in (1,2].
	\end{align}
\end{prop}

\begin{proof} \label{proof:QSR}
The achievability (i.e., the upper bound on $\eps$) of Quantum State Redistribution has been shown via convex splitting in Ref.~\cite{ADK+17} (see also \cite{AJW18, BCG25}).
Below, we will demonstrate how the sharpened convex splitting (Proposition~\ref{prop:splitting}) can improve the error bound using $\log M$ nats of noiseless classical communication. The idea of using convex splitting for Quantum State Redistribution is due to Anshu \textit{et al.} \cite{ADJ17}.
 
To begin the protocol, we let the sender (Alice) and the receiver (Bob) share $M$-copies of entanglement $\otimes_{m\in[M]} |\tau\rangle_{\bar{\A}_{m} \bar{\B}_{m}}$,
where Bob holds register $\bar{\B}_{m} \cong \C$  that purifies Alice's register $\bar{\A}_{m} \cong \C$, and $[M]:= \{1,2,\ldots, M\}$
We begin with the following pure state:
\begin{align} \label{eq:initial_bipartite}
|{\omega}\rangle := |\rho\rangle_{\R\A\C\B} 
\otimes_{m\in[M]} |\tau\rangle_{\bar{\A}_{m} \bar{\B}_{m}}.
\end{align}

Suppose, hopefully, by the protocol, we end up with the following pure state:
\begin{align} \label{eq:target_bipartite}
\begin{split} 
|\widehat{\omega}\rangle
&:= \frac{1}{\sqrt{M}} \sum_{m\in [M]} |m\rangle_{\mathsf{M}} 
|\rho\rangle_{\R\A\B\C_{m}} |0\rangle_{\bar{\A}_{m}}
\otimes_{\bar{m} \in [M]\setminus \{m\}} |\tau\rangle_{\bar{\A}_{\bar{m}} \bar{\B}_{\bar{m}}},
\end{split}
\end{align}
where Alice holds registers $\mathsf{M}$, $\A$, $\bar{\A}_{[M]} := \bar{\A}_{1} \bar{\A}_{2}\ldots \bar{\A}_{M}$,
Bob holds registers $\bar{\B}_{[M]} := \bar{\B}_{1} \bar{\B}_{2} \ldots \bar{\B}_{M}$, and notice that the register ${\C}_m \cong \bar{\B}_m$ for each $m\in[M]$ is held by Bob.
Alice measures her system $\mathsf{M}$, and sends the measurement outcome $m \in [M]$ to Bob via $\log M$ nats of classical communication.
At the receiver, Bob picks up the $m$-th register $\C_{m}$ to end up with $|\rho\rangle_{\R\A\B\C_{m}} \cong |\rho\rangle_{\R\A\B\C}$, which is exactly the target state we aimed for the Quantum State Splitting protocol.

Then, it remains to show that there exists a local operation protocol at Alice such that the desired state $|{\widehat{\omega}}\rangle$ in Eq.~\eqref{eq:target_bipartite} can be approximated via the Quantum State Redistribution protocol.
Note that the reduced state of the initial state $|{\omega}\rangle$ is
\begin{align}
{\omega}_{\R\B \bar{\B}_{[M]}} &= \rho_{\R\B}  \otimes_{m\in[M]} \tau_{\bar{\B}_{m}}. 
\end{align}
The convex splitting established in Proposition~\ref{prop:splitting} ensures that ${\omega}_{\R\B\bar{\B}_{[M]}}$ can be approximated  
by the following state
\begin{align} 
\begin{split}
\widehat{\omega}_{\R\B\bar{\B}_{[M]}}
&:= \frac{1}{M} \sum_{m\in [M] }
\rho_{\R\B\bar{\B}_{m}} 
\otimes_{\bar{m} \in [M]\setminus \{m\}} \tau_{\bar{\B}_{\bar{m}}},
\end{split}
\end{align}
within an error $\eps$ (in terms of purified distance) satisfying
\begin{align} \label{eq:error_QSR1}
\eps =   P\left(\widehat{\omega}_{\R\B\bar{\B}_{[M]}}, {\omega}_{\R\B \bar{\B}_{[M]}} \right)
    \leq \sqrt{ \frac{c_{\alpha-1}}{\alpha-1} } \cdot
    \e^{- \frac{(\alpha-1)}{2}\left[ \log M - {D}_{\alpha}(\rho_{\bar{B}\R\B} \Vert \tau_{\bar{B}} \otimes \rho_{\R\B} ) \right]  }, \quad \forall\, \alpha \in (1,2]
\end{align}
(by substituting registers $\A$ by $\bar{B}$ and $\B$ by $\R\B$).
Note that the sandwiched \Renyi divergence is invariant under isometry $\id_{\bar{B} \to \C}$ (since $\bar{B}\cong\C$); we can express the error bound in terms of ${D}_{\alpha}(\rho_{\C\R\B} \Vert \tau_{\C} \otimes \rho_{\R\B} )$. (Here, the register $\C$ is held by Alice or Bob is immaterial as it does not affect the divergence $\widetilde{D}_{\alpha}$.

Lastly, observe that $\widehat{\omega}_{\R\B \bar{\B}_{[M]}}$ is the reduced state of the desired pure state $|\widehat{\omega}\rangle$ given in Eq.~\eqref{eq:target_bipartite}. 
Hence, by Uhlmann's theorem (see Lemma~\ref{lemm:Uhlmann}), there exists an isometry $\mathscr{V}$ acting on register $\A\C\bar{\A}_{[M]}$ to register $\mathsf{M}\A\bar{\A}_{[M]}$ such that $P(\mathscr{V} (|{\omega}\rangle\langle \omega |), |\widehat{\omega}\rangle\langle\widehat{\omega}|)  = P(\widehat{\omega}_{\R\B\bar{\B}_{[M]}}, {\omega}_{\R\B \bar{\B}_{[M]}} )$.
Moreover, since the isometry $V$ is a local operation acting only on Alice's registers, this constitutes the Quantum State Splitting protocol with an error $\eps$.
\end{proof}

\subsection{Quantum Information Decoupling} \label{sec:decoupling}

\begin{defn}[Catalytic quantum information decoupling via removing a subsystem] \label{defn:decoupling}
Let $\rho_{\A\E}$ be a quantum state.
The protocol aims to decouple quantum information in system $\A$ from system $\E$ with assistance of a catalytic system $\bar{\A}$.
\begin{enumerate}[1.]
	\item 
	Alice holds a quantum register $\A$ and a catalytic register $\bar{\A}$, and Eve holds a quantum register $\E$.
	
	\item 
    Alice is free to choose a state $\tau_{\bar{\A}}$ in the catalytic system $\bar{\A}$.
	
	\item Alice applies a local unitary $\mathscr{U}$ on her systems $\A\bar{\A}$ to end up with systems $\A_1 \A_2$ (i.e., $|\A\bar{\A}| = |\A_1\A_2|)$, and then remove the system $\A_2$ (via partial trace).
\end{enumerate}		
	An $(M, \eps)$ catalytic quantum information decoupling protocol for $\rho_{\A\E}$ is the existence of the catalytic system $\bar{\A}$, a state $\tau_{\bar{\A}}$ on it, and a unitary $\mathscr{U}_{\A\bar{\A}\to \A_1\A_2}$ satisfying $|\A_2| \leq M$ and 
	\begin{align}
	\inf_{\omega_{\A_1}} P\left( \Tr_{\A_2}\left[ \mathscr{U}_{\A\bar{\A}\to \A_1\A_2}(\rho_{\A\E}\otimes \tau_{\bar{\A}}) \right], \omega_{\A_1} \otimes \rho_{\E} \right) \leq \eps,
	\end{align}
	where infimum is over all states $\omega_{\A_1}$.
\end{defn}


\begin{prop} \label{prop:catalystic_decoupling}
Let $\rho_{\A\E}$ be a quantum state.
Following Definition~\ref{defn:decoupling}, there exists an $(M, \eps)$ catalytic quantum information decoupling protocol for $\rho_{\A\E}$ satisfying
\begin{align}
    \eps 
    &\leq \sqrt{\frac{c_{\alpha-1}}{\alpha-1}} \e^{- (\alpha-1)\left[ \log M - \frac12 \inf_{\tau_{\A}} {D}_{\alpha}\left(\rho_{\A\E} \Vert \tau_{\A} \otimes \rho_{\E} \right) \right] }, \quad \forall\, \alpha\in(1,2].
\end{align}
\end{prop}

\begin{proof}
The proof strategy follows in a similar way to that of \cite[Theorem 7]{LY21a}.
Our contribution is to employ a key technique of sharpened one-shot bound for convex splitting established in Proposition~\ref{prop:splitting}.
For completeness, we detail the proof below.

We first show that there exists an $(\sqrt{M},\eps)$ decoupling operation via random unitaries:
\begin{align}
\mathscr{R}_{\A\bar{A}}: X \mapsto \frac{1}{\sqrt{M}} \sum_{m=1}^{\sqrt{M}} \mathscr{U}_m(X),
\end{align}
where $\bar{\A} = \A_2 \A_3 \ldots \A_{\sqrt{M}}$, and each unitary $\mathscr{U}_m$ is a swap between system $\A_m$ with $\A_1 \cong \A$. 
Let the catalytic state be $\otimes_{m=2}^{\sqrt{M}} \tau_{\A_m}$.
We have
\begin{align}
\mathscr{R}_{\A\bar{\A}}\left( \rho_{\A\R} \otimes \otimes_{m=2}^{\sqrt{M}} \tau_{\A_m} \right)
= \frac{1}{M} \sum_{m=1}^{\sqrt{M}} \rho_{\A_m \R} \otimes_{\bar{m}\neq m} \tau_{\A_{\bar{m}}}.
\end{align}

Via the convex splitting established in Proposition~\ref{prop:splitting},
\begin{align}
P\left( \mathscr{R}_{\A\bar{\A}}\left( \rho_{\A\R} \otimes \otimes_{m=2}^{\sqrt{M}} \tau_{\A_m} \right), \rho_{\R} \otimes \tau_{\A}^{\otimes \sqrt{M}} \right)
&\leq \sqrt{\frac{c_{\alpha-1}}{\alpha-1}} \e^{- (\alpha-1)\left[ \frac12 \log M - \frac12 {D}_{\alpha}\left(\rho_{\A\R} \Vert \tau_{\A} \otimes \rho_{\R} \right) \right] }, \quad \forall\, \alpha\in(1,2].
\end{align}

Lastly, recall that the existence of a $(\sqrt{M},\eps)$ decoupling map via the above random unitaries is equivalent to the existence of an $({M},\eps)$ decoupling map by removing systems (Definition~\ref{defn:decoupling}) \cite[Proposition 6]{LY21a}, we conclude the proof.
\end{proof}

\subsection{Quantum Channel Simulation} \label{sec:simulation}


\begin{defn}[Entanglement-assisted quantum channel simulation]\label{defn:channel-simulation}
Let $\mathscr{N}_{\A\to \B}$ be a quantum channel.
\begin{enumerate}[1.]
	\item 
	Alice at the sender holds a quantum register $\A$, and Bob at the receiver holds a quantum register
	$\B$, and $\R$ is an inaccessible reference system.
	
	\item A resource of free entanglement is shared between Alice (holding registers $\bar{\A}$) and Bob (holding register $\bar{\B}$).
	
	\item Alice applies a local operation on her systems and sends $\log M$
	nats of classical information to the receiver.

	\item 
	Upon receiving the message, Bob applies a local operation on his own system.
\end{enumerate}		
	An $( M, \eps)$ quantum channel simulation protocol for $\mathscr{N}_{\A\to\B}$ with a \emph{fixed} pure input state $\theta_{\R\A}$ satisfies
	\begin{align}
	P\left( \widehat{\mathscr{N}}_{\A\to \B}(\theta_{\R\A}) , {\mathscr{N}}_{\A\to \B}(\theta_{\R\A}) \right) \leq \eps,
	\end{align}
	where $\widehat{\mathscr{N}}_{\A\to \B}(\theta_{\R\A})$ is the effectively resulting state from Alice's register $\A$ to Bob's register $\B$. 
	The $\log M$ denotes the classical communication costs in the channel simulation protocol.

    An $( M, \eps)$ quantum channel simulation protocol for $\mathscr{N}_{\A\to\B}$ with an \emph{arbitrary} pure input state satisfies
	\begin{align}
	\sup_{\theta_{\R\A}} P\left( \widehat{\mathscr{N}}_{\A\to \B}(\theta_{\R\A}) , {\mathscr{N}}_{\A\to \B}(\theta_{\R\A}) \right) \leq \eps,
	\end{align}
    where the supremum is taken over all pure states $\theta_{\R\A}$.
\end{defn}

\subsubsection{Channel Simulation With a Fixed Input} \label{sec:QSS}

\begin{prop} \label{prop:simulation_fixed}
Let $\mathscr{N}_{\A\to \B}$ be a quantum channel and let $\theta_{\R\A}$ be a fixed pure input state.
Following Definition~\ref{defn:channel-simulation}, there exists an $( M, \eps)$ quantum channel simulation protocol for $\mathscr{N}_{\A\to\B}$ with the input state $\theta_{\R\A}$ satisfying
	\begin{align} \label{eq:error_simulation_fixed}
	\eps &\leq
	\sqrt{ \frac{c_{\alpha-1}}{\alpha-1} } \cdot
    \e^{- \frac{(\alpha-1)}{2}\left[ \log M - \inf_{\tau_{\B}} {D}_{\alpha}(\rho_{\B\R} \Vert \tau_{\B} \otimes \rho_{\R} ) \right]  }, \quad \forall\, \alpha \in (1,2],
	\end{align}
    where $\rho_{\R\B} := \mathscr{N}_{\A\to\B}(\theta_{\R\A})$ and the infimum on the right-hand side is over all states $\tau_{\B}$.
\end{prop}

Proposition~\ref{prop:simulation_fixed} improves on \cite[Proposition 13]{LY21b} by a factor $c_{\alpha-1} \in [1/2,1]$ for $\alpha \in (1,2]$ and by removing the dimension-dependent factor $|\text{spec}(\mathcal{H}_{\B})|^{\alpha-1}$.

\begin{proof}
Let $\mathscr{U}_{\A\to\B\E}$ be a Stinespring dilation of $\mathscr{N}_{\A\to\B}$.
Alice first simulates a local isometry $\mathscr{U}_{\A\to\B\E}$ at her side to obtain the state
\begin{align}
\rho_{\R\E\B} :=
\left(\mathscr{U}_{\A\to\B\E} \otimes \id_{\R \bar{\R}} \right) \theta_{\A\R}.
\end{align}
Next, we apply the Quantum State Splitting for $\rho_{\R\E\B}$ given in Section~\ref{sec:QSR} with registers 
$\R\leftarrow \R$ at the reference system, 
$\A \leftarrow \E$ at Alice, 
and $\C \leftarrow \B$ at Bob (i.e., a Quantum State Redistribution with the register $\B$ being void), to send the channel output system $\B$ to Bob via $\log M$ nats of classical communication and $M$-copies of $|\tau\rangle_{\bar{\A}\bar{\B}}$, where $\bar{B}\cong \B$.

Let the overall resulting state be $\mathscr{U}_{\A\to\B\E}(\theta_{\A\R})$. We obtain the following error bound by Proposition~\ref{prop:QSR}:
\begin{align}
    P\left( \hat{\mathscr{U}}_{\A\to\B\E}(\theta_{\A\R}), \mathscr{U}_{\A\to\B\E}(\theta_{\A\R}) \right)
    &\leq \sqrt{ \frac{c_{\alpha-1}}{\alpha-1} } \cdot
    \e^{- \frac{(\alpha-1)}{2}\left[ \log M - {D}_{\alpha}(\rho_{\B\R} \Vert \tau_{\B} \otimes \rho_{\R} ) \right]  }, \quad \forall\, \alpha \in (1,2].
\end{align}
Lastly, by tracing out the systems $\E$, the data-processing inequality of purified distance $P(\cdot,\cdot)$, and optimizing over all shared entangled states $|\tau\rangle_{\bar{\A}\bar{\B}}$, we conclude the proof.
\end{proof}

\subsubsection{Channel Simulation With Arbitrary Inputs} \label{sec:simulation_arbitrary}

\begin{prop} \label{prop:simulation_arbitrary}
Let $\mathscr{N}_{\A\to \B}$ be a quantum channel.
Following Definition~\ref{defn:channel-simulation}, there exists an $( M, \eps)$ quantum channel simulation protocol for $\mathscr{N}_{\A\to\B}$ with arbitrary pure input states satisfying
	\begin{align} \label{eq:error_simulation_arbitrary}
	\eps &\leq
	\sqrt{ \frac{c_{\alpha-1}}{\alpha-1} } \cdot
    \e^{- \frac{(\alpha-1)}{2}\left[ \log M - \sup_{\theta_{\R\A}} \inf_{\tau_{\B}} {D}_{\alpha}(\rho_{\B\R} \Vert \tau_{\B} \otimes \rho_{\R} ) \right]  }, \quad \forall\, \alpha \in (1,2],
	\end{align}
    where $\rho_{\R\B} := \mathscr{N}_{\A\to\B}(\theta_{\R\A})$, 
    the supremum is over all pure input states $\theta_{\R\A}$, and
    the infimum is over all states $\tau_{\B}$.
\end{prop}

Proposition~\ref{prop:simulation_arbitrary} improves on \cite[Theorem 9]{LY21b} by a factor $c_{\alpha-1} \in [1/2,1]$ for $\alpha \in (1,2]$ and by removing the dimension-dependent factor $ |\text{spec}(\mathcal{H}_{\A})|^2 \cdot
|\text{spec}(\mathcal{H}_{\R})|^{\alpha-1} \cdot |\text{spec}(\mathcal{H}_{\B})|^{\alpha-1}$.

In the $n$-fold independent and identical setting where $\mathscr{N}_{\A\to\B}\leftarrow \mathscr{N}_{\A\to\B}^{\otimes n}$ and $M = \exp(nR)$ with $R > D(\rho_{\B\R}\Vert \rho_{\B} \otimes \rho_{\R})$,
the additivity property (see Lemma~\ref{lemm:sandwiched_channel_additive} below) and Proposition~\ref{prop:simulation_arbitrary} yield an achievable (regularized) error exponent:
\begin{align}
\sup_{\alpha \in (1,2]} \frac{(\alpha-1)}{2}\left[ R - \sup_{\theta_{\R\A}} \inf_{\tau_{\B}} \widetilde{D}_{\alpha}(\rho_{\B\R} \Vert \tau_{\B} \otimes \rho_{\R} ) \right],
\end{align}
which has been shown to be tight for $R < \left.\frac{\d}{\d s} s \inf_{\tau_{\B}} \widetilde{D}_{1+s}(\rho_{\B\R}\Vert\tau_{\B}\otimes \rho_{\R})\right|_{s=1}$ \cite[Theorem 11]{LY21b}.

Before commencing the proof, let us add some historical remarks.
Quantum channel simulation with arbitrary input states has been extensively studied in the literature \cite{BSS+02,BDH+14,BCR11,LY21b, BCG25}.
Preliminary methods of handling arbitrary input states rely on the so-called \emph{post-selection technique} \cite{ChristKoenRennerPostSelect}, which is also known as the \emph{de Finetti reduction}.
A recent work \cite{CJT24} proposed the idea of using the minimax identity to bypass the post-selection technique for bounding the minimal communication cost.
Below, we demonstrate that the minimax identity is also useful in the error bound by resorting to a concavity and convexity properties of the sandwiched \Renyi divergence (Lemmas~\ref{lemm:concavity_state} and \ref{lemm:convexity_order}).

\begin{proof}[Proof of Proposition~\ref{prop:simulation_arbitrary}]\label{proof:simulation_arbitrary}
By Definition~\ref{defn:channel-simulation}, we would like to show the existence of an $(M,\eps)$-quantum simulation protocol such that for all pure input states $\theta_{\R\A}$, the purified distance is at most $\eps$.
Thanks to the minimax identity to interchange the supremum between $\theta_{\R\A}$ and infimum between all protocols (Lemma~\ref{lemm:minimax}), it is sufficient to show the error bound for each input states $\theta_{\R\A}$, and then choose the worst input in the end.

For any pure input state $\theta_{\R\A}$, we apply Proposition~\ref{prop:simulation_fixed} to simulate a channel $\hat{\mathscr{N}}_{\A\to\B}$ with an error
\begin{align}
P\left( \widehat{\mathscr{N}}_{\A\to \B}(\theta_{\R\A}) , {\mathscr{N}}_{\A\to \B}(\theta_{\R\A}) \right) 
&\leq
	\sqrt{ \frac{c_{\alpha-1}}{\alpha-1} } \cdot
    \e^{- \frac{(\alpha-1)}{2}\left[ \log M - \inf_{\tau_{\B}} {D}_{\alpha}(\rho_{\B\R} \Vert \tau_{\B} \otimes \rho_{\R} ) \right]  }, \quad \forall\, \alpha \in (1,2],
	\end{align}
where $\rho_{\R\B} := \mathscr{N}_{\A\to\B}(\theta_{\R\A})$.

Then, we maximize over all pure input states $\theta_{\R\A}$ (i.e., the worst case scenario).
By invoking the convexity of the map $\alpha \mapsto (\alpha-1) \inf_{\tau_{\B}} \widetilde{D}_{\alpha} (\rho_{\B\R} \Vert \tau_{\B} \otimes \rho_{\R})$ on the convex set $\alpha \in (1,2]$ (Lemma~\ref{lemm:convexity_order}),
the concavity of the map $\theta_{\R\A} \mapsto \inf_{\tau_{\B}} \widetilde{D}_{\alpha} (\rho_{\B\R} \Vert \tau_{\B} \otimes \rho_{\R})$ (Lemma~\ref{lemm:concavity_state})\footnote{The input state $\theta_{\R\A}$ can be taken to be general mixed state, but it can be attained by pure states due to the convexity. Hence, without loss of generality, we only consider pure input states in Definition~\ref{defn:channel-simulation}},
and Sion's minimax theorem, we have
\begin{align}
\inf_{\theta_{\R\A}} \sup_{\alpha \in (1,2]} (1-\alpha) \inf_{\tau_{\B}} \widetilde{D}_{\alpha}(\rho_{\B\R} \Vert \tau_{\B} \otimes \rho_{\R} )
&=  \sup_{\alpha \in (1,2]} \inf_{\theta_{\R\A}} (1-\alpha) \inf_{\tau_{\B}} \widetilde{D}_{\alpha}(\rho_{\B\R} \Vert \tau_{\B} \otimes \rho_{\R} ),
\end{align}
which concludes the proof.
\end{proof}

\section*{Acknowledgments}
PL and HC are supported under grants 113-2119-M-007-006, 113-2119-M-001-006, NSTC 114-2124-M-002-003, NTU-113V1904-5, NTU-CC-113L891605, NTU-
113L900702, NTU-114L900702, and NTU-114L895005.
LG is supported in part by National Natural Science Foundation of China (12401163, 62171212).
CH received funding by the Deutsche Forschungsgemeinschaft (DFG, German
Research Foundation) – 550206990.


\appendix

\section{Auxiliary Lemmas} \label{sec:lemmas}

\begin{lemm}[Finite-rank approximations, relative entropy and sandwiched quasi divergence {\cite[Lemmas 3 \& 4]{Lin74}}, {\cite[Corollary 5.12]{Petz_book_1993}}, {\cite[Proposition~III.39]{Mos23}}] \label{lemm:finit-rank}
Let $A$ and $B$ be non-zero trace-class operators on an infinite-dimensional separable Hilbert space.
Then, for any $\alpha >1 $,
\begin{align}
D(A \Vert B) 
&\coloneq \Tr\left[ A (\log A - \log B) + B-A \right]
= \lim_{n\to\infty} D(P_n A P_n \Vert P_n B P_n),
\\
\widetilde{Q}_{\alpha}(A \Vert B) 
&\coloneq 
\Tr\left[ \left( \sigma^{-\frac{1}{2\alpha}} \rho  \sigma^{-\frac{1}{2\alpha}} \right)^{\alpha} \right]
= \lim_{n\to\infty} \widetilde{Q}_{\alpha}(P_n A P_n \Vert P_n B P_n)
\end{align}
where $(P_n)_{n\in\mathds{N}}$ is any sequence of projections such that $\Tr[P_n] = n$, $P_{n-1} \leq P_n$, and $P_n \nearrow \I$ in the strong operator topology.
\end{lemm}
[Go back to \hyperref[proof:sharp_one-shot]{\textit{Proof of Theorem~\ref{theo:sharp_one-shot}}}]

\medskip

\begin{lemm}[Finite-rank approximations, integral quasi divergence] \label{lemm:finit-rank-integral}
Let $A$ and $B$ be non-zero trace-class operators on an infinite-dimensional separable Hilbert space.

Then, for any $\alpha >1 $,
\begin{align}
{Q}_{\alpha}(A \Vert B) 
&\coloneq (\alpha-1) \int_0^\infty \Tr\left[ \left( \left( B + t\I\right)^{-\nicefrac{1}{2}} A \left( B + t\I\right)^{-\nicefrac{1}{2}} \right)^{\alpha} \right]\d t
= \lim_{n\to\infty} {Q}_{\alpha}(P_n A P_n \Vert P_n B P_n)
\end{align}
where $(P_n)_{n\in\mathds{N}}$ is a sequence of projections of $B$ such that $\Tr[P_n] = n$, $P_{n-1} \leq P_n$, and $P_n \nearrow \I$ in the strong operator topology.
\end{lemm}
[Go back to \hyperref[proof:sharp_one-shot]{\textit{Proof of Theorem~\ref{theo:sharp_one-shot}}}]

\begin{proof}[Proof of Lemma~\ref{lemm:finit-rank-integral}]
Fix $s>0$. Assume
\begin{align}
    \int_0^\infty \Tr\left[ \left( (B+t\I)^{-\nicefrac{1}{2}} A (B+t\I)^{-\nicefrac{1}{2}} \right)^{1+s} \right] \, \d t < \infty.
\end{align}
Denote 
\begin{align}
    f(t)
    &= \int_0^\infty \Tr\left[ \left( (B+t\I)^{-\nicefrac{1}{2}} A (B+t\I)^{-\nicefrac{1}{2}} \right)^{1+s} \right] \, \d t,
    \\
    f_n(t)
    &= \int_0^\infty \Tr\left[ \left( (B+t\I)^{-\nicefrac{1}{2}} P_n A P_n (B+t\I)^{-\nicefrac{1}{2}} \right)^{1+s} \right] \, \d t
    \\
    &= \int_0^\infty \Tr\left[ \left( (P_n B P_n +t\I)^{-\nicefrac{1}{2}} P_n A P_n ( P_n B P_n +t\I)^{-\nicefrac{1}{2}} \right)^{1+s} \right] \, \d t.
\end{align}
where $P_n$ is the spectral projection for $B$ and $P_n \nearrow \I$.
Since $t\mapsto f(t)$ is decreasing, we have $f(t) < \infty$ for all $t>0$.

It suffices to show $f_n \leq f$ and $f_n \to f$ pointwise.
Then, by Monotone Convergence Theorem,
\begin{align}
    \int_0^\infty f_n(t) \, \d t
    \nearrow \int_0^\infty f(t) \, \d t.
\end{align}

We first show $f(t) \geq f_n (t)$ for all $t>0$.
Define the channel $\Phi_n (X) = P_n^\perp X P_n^{\perp} + P_n X P_n$.
Note that $\Phi_n$ is unital and trace-preserving.
Hence, the Schatten $p$-norm is contractive under $\Phi_n$, i.e.,
\begin{align}
\left\| \Phi_n(X) \right\| \leq \|X\|_p, \quad \forall\, p\geq 1.
\end{align}
Then,
\begin{align}
    f_n(t)
    &= \Tr\left[ \left(  \frac{P_n A P_n}{ P_n B P_n + t \I } \right)^{1+s} \right]
    \\
    &= \Tr\left[ \left( P_n (B+t\I)^{-\nicefrac{1}{2}} A (B+t\I)^{-\nicefrac{1}{2}}  P_n \right)^{1+s} \right]
    \\
    &\leq \Tr\left[ \left( P_n (B+t\I)^{-\nicefrac{1}{2}} A (B+t\I)^{-\nicefrac{1}{2}}  P_n \right)^{1+s} \right]
    + \Tr\left[ \left( P_n^\perp (B+t\I)^{-\nicefrac{1}{2}} A (B+t\I)^{-\nicefrac{1}{2}}  P_n^\perp \right)^{1+s} \right]
    \\
    &= \left\| \Phi_n \left( (B+t\I)^{-\nicefrac{1}{2}} A (B+t\I)^{-\nicefrac{1}{2}} \right) \right\|_{1+s}^{1+s}
    \\
    &\leq 
    \left\| (B+t\I)^{-\nicefrac{1}{2}} A (B+t\I)^{-\nicefrac{1}{2}} \right\|_{1+s}^{1+s}
    \\
    &= f(t).
\end{align}

Next, by using the fact that $P_n$ is a spectral projection of $B$, we calculate
\begin{align}
f_n(t)^{\frac{1}{1+s}} - f(t)^{\frac{1}{1+s}} 
&= \left\| (B+t\I)^{-\nicefrac{1}{2}} P_n A P_n (B+t\I)^{-\nicefrac{1}{2}} \right\|_{1+s}
- \left\| (B+t\I)^{-\nicefrac{1}{2}} A (B+t\I)^{-\nicefrac{1}{2}} \right\|_{1+s}
\\
&\leq \left\| (B+t\I)^{-\nicefrac{1}{2}} \right\|_{\infty} \left\| P_n A P_n - A \right\|_{1+s} \left\| (B+t\I)^{-\nicefrac{1}{2}} \right\|_{1+s}
\\
&\leq t^{-1} \left\| P_n A P_n - A \right\|_1.
\end{align}
Then,
\begin{align}
    \left\| P_n A P_n - A \right\|_1
    &\leq \left\| P_n - A \right\|_1 
    + \left\| P_n A P_n - P_n A \right\|_1
    \\
    &\leq \left\| P_n A^{-\nicefrac{1}{2}} - A^{\nicefrac{1}{2}} \right\|_2
    \left\|A^{\nicefrac{1}{2}} \right\|_2
    + \left\| P_n A^{\nicefrac{1}{2}}  \right\|_2 \left\|  A^{-\nicefrac{1}{2}} P_n  - A^{\nicefrac{1}{2}} \right\|_2
    \\
    &= \sqrt{ \Tr\left[ P_n A P_n + A - P_n A - A P_n  \right]  \Tr\left[ A \right] }
    + \sqrt{ \Tr\left[ P_n A \right]\Tr\left[ P_n A P_n + A - P_n A - A P_n  \right]   }
    \\
    &\rightarrow 0,
\end{align}
as $\lim_{n\to\infty} \Tr\left[ P_n A P_n \right] = \lim_{n\to 0} \Tr[P_n A ] = \Tr[A]$.
Hence, $f_n \to f$ pointwise.
We complete the proof.
\end{proof}

\medskip

\begin{lemm}[Uhlmann's theorem \cite{Uhl76}] \label{lemm:Uhlmann}
	Let $\psi_{\A\B} = |\psi\rangle\langle \psi|_{\A\B}$ and $\varphi_{\A\C} = |\varphi\rangle\langle \varphi|_{\A\C}$ be two pure quantum states.
	Then, there exists an isometry $\mathscr{V}_{\B\to \C}$ satisfying
    \begin{align}
        P\left( \psi_{\A} , \varphi_{\A} \right)=
		P\left(\mathscr{V}_{\B\to C}( \psi_{\A\B} ) , \varphi_{\A\C}\right)  .
    \end{align}
		
\end{lemm}
[Go back to \hyperref[proof:QSR]{\textit{Proof of Proposition~\ref{prop:QSR}}}]

\medskip

\begin{lemm}[A Minimax identity {\cite{CJT24}}] \label{lemm:minimax}
Let $\mathfrak{P}_{\A\to\B}^{(M)}$ be the set of all quantum simulation protocols for the channel in Definition~\ref{defn:channel-simulation}, i.e.,~all entanglement-assisted local operations and $\log M$ nats of one-way classical communication.
Then, for any quantum channel $\mathscr{N}_{\A\to\B}$,
\begin{align}
\inf_{\hat{\mathscr{N}}_{\A\to\B} \in \mathfrak{P}_{\A\to\B}^{(M)}} \sup_{\theta_{\R\A}} P\left( \widehat{\mathscr{N}}_{\A\to \B}(\theta_{\R\A}) , {\mathscr{N}}_{\A\to \B}(\theta_{\R\A}) \right)
= \sup_{\theta_{\R\A}} \inf_{\hat{\mathscr{N}}_{\A\to\B} \in \mathfrak{P}_{\A\to\B}^{(M)}}  P\left( \widehat{\mathscr{N}}_{\A\to \B}(\theta_{\R\A}) , {\mathscr{N}}_{\A\to \B}(\theta_{\R\A}) \right).
\end{align}
\end{lemm}
[Go back to \hyperref[proof:simulation_arbitrary]{\textit{Proof of Proposition~\ref{prop:simulation_arbitrary}}}]
\medskip

\begin{lemm}[Additivity {\cite[Lemma 5]{GW14}}] \label{lemm:sandwiched_channel_additive}
Let $\mathscr{N}: \A \to \B$ be a quantum channel.
For any integer $n$, let $\rho_{\R^n\B^n} := \mathscr{N}_{\A\to\B}^{\otimes n}(\theta_{\R^n\A^n})$ for any pure state $\theta_{\R^n\A^n}$.
Then, for all $\alpha > 1$,
\begin{align}
\sup_{\theta_{\R^n\A^n}} \inf_{\sigma_{\B^n}} \widetilde{D}_{\alpha}\left(\rho_{\R^n\B^n} \Vert \rho_{\R^n} \otimes \sigma_{\B^n} \right)
= n \cdot \sup_{\theta_{\R\A}} \inf_{\sigma_{\B}} \widetilde{D}_{\alpha}\left(\rho_{\R\B} \Vert \rho_{\R} \otimes \sigma_{\B} \right).
\end{align}
\end{lemm}
\noindent [Go back to \hyperref[proof:simulation_arbitrary]{\textit{Proof of Proposition~\ref{prop:simulation_arbitrary}}}]
\medskip

\begin{lemm}[Concavity in state {\cite[Lemma 4]{GW14}}] \label{lemm:concavity_state}
Let $\mathscr{N}: \A \to \B$ be a quantum channel and let $\rho_{\R\B} := \mathscr{N}_{\A\to\B}(\theta_{\R\A})$ for any state $\theta_{\R\A}$.
Then, the map
\begin{align}
\theta_{\R\A} &\mapsto \inf_{\sigma_{\B}} \widetilde{D}_{\alpha} (\rho_{\R\B} \Vert \rho_{\R} \otimes \sigma_{\B})
\end{align}
is concave for all states on $\mathcal{H}_{\R}\otimes \mathcal{H}_{\A}$.
\end{lemm}
\noindent [Go back to \hyperref[proof:simulation_arbitrary]{\textit{Proof of Proposition~\ref{prop:simulation_arbitrary}}}]
\medskip

\begin{lemm}[Convexity in order] \label{lemm:convexity_order}
Let $\rho_{\A\B}$ be a state.
Then, the map
\begin{align}
\alpha &\mapsto (\alpha-1) \inf_{\sigma_{\B}} \widetilde{D}_{\alpha} (\rho_{\A\B} \Vert \rho_{\A} \otimes \sigma_{\B})
\end{align}
is convex for $\alpha > 1$.
\end{lemm}
\noindent [Go back to \hyperref[proof:simulation_arbitrary]{\textit{Proof of Proposition~\ref{prop:simulation_arbitrary}}}]

\begin{proof}[Proof of Lemma~\ref{lemm:convexity_order}] \label{proof:convexity_order}
The case of system $\A$ being classical has been shown in \cite[Theorem 11]{CGH19}. In the following, we adopt a similar proof technique via complex interpolation theory.

For all $\alpha\geq 1$, we let $\alpha' = \frac{\alpha}{\alpha-1}$ be its H\"older conjugate.
Fix $\alpha = (1-\theta) \alpha_0 + \theta \alpha_1$, $\theta \in [0,1]$, and $\alpha_0, \alpha_1 \geq 1$.
By the definition of the sandwiched \Renyi divergence, we write
\begin{align}
\inf_{\sigma_{\B}} \widetilde{D}_{\alpha} \left( \rho_{\A\B} \Vert \rho_{\A} \otimes \sigma_{\B} \right)
&= \frac{\alpha}{\alpha-1} \log \left\| \left( \rho_{\A}^{- \frac{1}{2\alpha'}} \otimes \sigma_{\B}^{- \frac{1}{2\alpha'}} \right) \rho_{\A\B} \left( \rho_{\A}^{- \frac{1}{2\alpha'}} \otimes \sigma_{\B}^{- \frac{1}{2\alpha'}} \right) \right\|_{\alpha}
\\
&=: \frac{\alpha}{\alpha-1} \log \left\|  \rho_{\A}^{- \frac{1}{2\alpha'}} \rho_{\A\B} \rho_{\A}^{- \frac{1}{2\alpha'}} \right\|_{S_1(\B, S_\alpha(\A))}, 
\end{align}
by using the notation of the amalgamated norm \cite{JP10}.
The desired convexity is then equivalent to 
\begin{align}
\left\|  \rho_{\A}^{- \frac{1}{2\alpha'}} \rho_{\A\B} \rho_{\A}^{- \frac{1}{2\alpha'}} \right\|_{S_1(\B, S_\alpha(\A))}
\leq 
\left\|  \rho_{\A}^{- \frac{1}{2\alpha_0'}} \rho_{\A\B} \rho_{\A}^{- \frac{1}{2\alpha_0'}} \right\|_{S_1(\B, S_{\alpha_0}(\A))}^{\frac{\alpha_0(1-\theta)}{\alpha}}
\left\|  \rho_{\A}^{- \frac{1}{2\alpha_1'}} \rho_{\A\B} \rho_{\A}^{- \frac{1}{2\alpha_1'}} \right\|_{S_1(\B, S_{\alpha_1}(\A))}^{\frac{\alpha_1\theta}{\alpha}}.
\end{align}

Denote $y = \frac{\alpha_1 \theta}{\alpha}$ and $1-y = \frac{\alpha_0(1-\theta)}{\alpha}$ such that $\frac{1}{\alpha} = \frac{1-y}{\alpha_0} + \frac{y}{\alpha_1}$.
We consider an analytic family of operators:
\begin{align}
    F: z \mapsto \mathtt{N}^y \mathtt{M}^{1-y} \rho_{\A}^{-\frac12(1-\frac{1-z}{\alpha_0}-\frac{z}{\alpha_1})} \rho_{\A\B} \rho_{\A}^{-\frac12(1-\frac{1-z}{\alpha_0}-\frac{z}{\alpha_1})},
\end{align}
where
\begin{align}
    \mathtt{M}
    &:= \left\| \rho_{\A}^{-\frac{1}{2\alpha_0'}} \rho_{\A\B} \rho_{\A}^{-\frac{1}{2\alpha_0'}} \right\|_{S_1(\B,S_{\alpha_0}(\A))}^{-1},
    \\
    \mathtt{N}
    &:= \left\| \rho_{\A}^{-\frac{1}{2\alpha_1'}} \rho_{\A\B} \rho_{\A}^{-\frac{1}{2\alpha_1'}} \right\|_{S_1(\B,S_{\alpha_1}(\A))}^{-1}.
\end{align}
We bound the boundary the map $F$ as follows:
\begin{align}
\left\|F(\mathrm{i}t)\right\|_{S_1(\B,S_{\alpha_0}(\A))} &= \left\| \mathtt{M}^{-\mathrm{i}t} \mathtt{N}^{\mathrm{i}t} \mathtt{N}^{-1} \rho_{\A}^{-\frac{1}{2\alpha_0}} \rho_{\A\B} \rho_{\A}^{-\frac{1}{2\alpha_0}} \right\|_{S_1(\B,S_{\alpha_0}(\A))}
\leq 1,
\\
\left\|F(1+\mathrm{i}t)\right\|_{S_1(\B,S_{\alpha_1}(\A))} &= \left\| \mathtt{M}^{\mathrm{i}t} \mathtt{N}^{-\mathrm{i}t} \mathtt{M}^{-1} \rho_{\A}^{-\frac{1}{2\alpha_1}} \rho_{\A\B} \rho_{\A}^{-\frac{1}{2\alpha_1}} \right\|_{S_1(\B,S_{\alpha_1}(\A))}
\leq 1.
\end{align}

By interpolation (Lemma~\ref{lemm:interpolation}) and the fact that  the amalgamated norm forms an interpolation spaces~\cite{JP10},
\begin{align}
    1 
    &\geq \left\|F(y)\right\|_{S_1(\B,S_{\alpha}(\A))}
    \\
    &= \left\| \mathtt{N}^y \mathtt{M}^{1-y} \rho_{\A}^{-\frac{1}{2\alpha'}} \rho_{\A\B} \rho_{\A}^{-\frac{1}{2\alpha'}} \right\|_{S_1(\A,S_{\alpha}(\B))}
    \\
    &= \mathtt{N}^y \mathtt{M}^{1-y} \left\|  \rho_{\A}^{-\frac{1}{2\alpha'}} \rho_{\A\B} \rho_{\A}^{-\frac{1}{2\alpha'}} \right\|_{S_1(\A,S_{\alpha}(\B))},
\end{align}
which translates to the desired convexity:
\begin{align}
\left\|  \rho_{\A}^{-\frac{1}{2\alpha'}} \rho_{\A\B} \rho_{\A}^{-\frac{1}{2\alpha'}} \right\|_{S_1(\A,S_{\alpha}(\B))}
\leq \mathtt{M}^{y-1} \mathtt{N}^{-y} .
\end{align}
\end{proof}

\begin{lemm}[Riesz--Thorin interpolation theorem  {\cite{BL76}}] \label{lemm:interpolation} Let $(X_0, X_1)$ and $(Y_0, Y_1)$ be two compatible couples of Banach spaces and let $(X_0, X_1)_\theta$ and $(Y_0, Y_1)_\theta$ be the corresponding complex interpolation space of exponent $\theta\in [0,1]$. Suppose $T:X_0 + X_1 \to Y_0 + Y_1$, is a linear operator bounded from $X_j$ to $Y_j$, $j=0, 1$. Then $T$ is bounded from $(X_0, X_1)_\theta$ to $(Y_0, Y_1)_\theta$, and moreover,
\[ \|T:(X_0, X_1)_\theta\to (Y_0, Y_1)_\theta\|_{}\leq \|T:X_0\to Y_0\|_{}^{1-\theta }\|T:X_1\to Y_1\|_{}^{\theta}\ ,\  \theta\in [0,1].\]
\end{lemm}
\noindent [Go back to \hyperref[proof:convexity_order]{\textit{Proof of Lemma~\ref{lemm:convexity_order}}}]


{\larger
\bibliographystyle{IEEEtran}
\bibliography{reference3.bib, operator3.bib}

\begin{thebibliography}{10}
\providecommand{\url}[1]{#1}
\csname url@samestyle\endcsname
\providecommand{\newblock}{\relax}
\providecommand{\bibinfo}[2]{#2}
\providecommand{\BIBentrySTDinterwordspacing}{\spaceskip=0pt\relax}
\providecommand{\BIBentryALTinterwordstretchfactor}{4}
\providecommand{\BIBentryALTinterwordspacing}{\spaceskip=\fontdimen2\font plus
\BIBentryALTinterwordstretchfactor\fontdimen3\font minus
  \fontdimen4\font\relax}
\providecommand{\BIBforeignlanguage}[2]{{%
\expandafter\ifx\csname l@#1\endcsname\relax
\typeout{** WARNING: IEEEtran.bst: No hyphenation pattern has been}%
\typeout{** loaded for the language `#1'. Using the pattern for}%
\typeout{** the default language instead.}%
\else
\language=\csname l@#1\endcsname
\fi
#2}}
\providecommand{\BIBdecl}{\relax}
\BIBdecl

\bibitem{Car09}
E.~Carlen, ``Trace inequalities and quantum entropy: an introductory course,''
  in \emph{Contemporary Mathematics}.\hskip 1em plus 0.5em minus 0.4em\relax
  American Mathematical Society ({AMS}), 2010, vol. 529, pp. 73--140.

\bibitem{Lieb_book_2002}
M.~Loss and M.~B. Ruskai, \emph{Inequalities: Selecta of Elliott H.
  Lieb}.\hskip 1em plus 0.5em minus 0.4em\relax Springer Berlin Heidelberg,
  2002.

\bibitem{Simon_book_2019}
B.~Simon, \emph{Loewner’s Theorem on Monotone Matrix Functions}.\hskip 1em
  plus 0.5em minus 0.4em\relax Springer International Publishing, 2019.

\bibitem{Lieb_book_2022}
R.~L. Frank, A.~Laptev, M.~Lewin, and R.~Seiringer, \emph{The Physics and
  Mathematics of Elliott Lieb}.\hskip 1em plus 0.5em minus 0.4em\relax EMS
  Press, June 2022.

\bibitem{MDS+13}
M.~M{\"u}ller-Lennert, F.~Dupuis, O.~Szehr, S.~Fehr, and M.~Tomamichel, ``On
  quantum {R{\'e}nyi} entropies: A new generalization and some properties,''
  \emph{Journal of Mathematical Physics}, vol.~54, no.~12, p. 122203, 2013.

\bibitem{WWY14}
M.~M. Wilde, A.~Winter, and D.~Yang, ``Strong converse for the classical
  capacity of entanglement-breaking and {Hadamard} channels via a sandwiched
  {R{\'{e}}nyi} relative entropy,'' \emph{Communications in Mathematical
  Physics}, vol. 331, no.~2, pp. 593--622, Jul 2014.

\bibitem{hirche2023quantum}
C.~Hirche and M.~Tomamichel, ``Quantum r\'enyi and $f$-divergences from
  integral representations,'' \emph{Communications in Mathematical Physics},
  vol. 405, no.~9, p. 208, 2024.

\bibitem{beigi2025some}
S.~Beigi, C.~Hirche, and M.~Tomamichel, ``Some properties and applications of
  the new quantum $ f $-divergences,'' \emph{arXiv preprint arXiv:2501.03799},
  2025.

\bibitem{preparation2}
\BIBentryALTinterwordspacing
P.-C. Liu, C.~Hirche, and H.-C. Cheng, ``Layer cake representations for quantum
  divergences,'' 2025, arXiv:2507.07065 [quant-ph]. [Online]. Available:
  \url{http://arxiv.org/abs/2507.07065}
\BIBentrySTDinterwordspacing

\bibitem{preparation}
\BIBentryALTinterwordspacing
H.-C. Cheng and P.-C. Liu, ``Error exponents for quantum packing problems via
  an operator layer cake theorem,'' 2025, arXiv:2507.06232 [quant-ph].
  [Online]. Available: \url{http://arxiv.org/abs/2507.06232}
\BIBentrySTDinterwordspacing

\bibitem{Mos23}
M.~Mosonyi, ``The strong converse exponent of discriminating
  infinite-dimensional quantum states,'' \emph{Communications in Mathematical
  Physics}, vol. 400, no.~1, pp. 83--132, March 2023.

\bibitem{Ume56}
H.~Umegaki, ``Conditional expectation in an operator algebra, {II},''
  \emph{Tohoku Mathematical Journal}, vol.~8, no.~1, pp. 86--100, jan 1956.

\bibitem{Hay15}
M.~Hayashi, ``Quantum wiretap channel with non-uniform random number and its
  exponent and equivocation rate of leaked information,'' \emph{{IEEE}
  Transactions on Information Theory}, vol.~61, no.~10, pp. 5595--5622, oct
  2015.

\bibitem{CG22}
H.-C. Cheng and L.~Gao, ``Error exponent and strong converse for quantum soft
  covering,'' \emph{IEEE Transactions on Information Theory}, vol.~70, no.~5,
  pp. 3499--3511, May 2024.

\bibitem{SGC22b}
Y.-C. Shen, L.~Gao, and H.-C. Cheng, ``Optimal second-order rates for quantum
  soft covering and privacy amplification,'' \emph{IEEE Transactions on
  Information Theory}, vol.~70, no.~7, pp. 5077--5091, July 2024.

\bibitem{HCG24}
\BIBentryALTinterwordspacing
M.~Hayashi, H.-C. Cheng, and L.~Gao, ``Resolvability of classical-quantum
  channels,'' 2024. [Online]. Available: \url{https://arxiv.org/abs/2410.16704}
\BIBentrySTDinterwordspacing

\bibitem{TPM17}
M.~Bastani~Parizi, E.~Telatar, and N.~Merhav, ``Exact random coding secrecy
  exponents for the wiretap channel,'' \emph{IEEE Transactions on Information
  Theory}, vol.~63, no.~1, p. 509–531, January 2017.

\bibitem{Hay15_PA}
M.~Hayashi, ``Precise evaluation of leaked information with secure randomness
  extraction in the presence of quantum attacker,'' \emph{Communications in
  Mathematical Physics}, vol. 333, no.~1, pp. 335--350, 2015.

\bibitem{KL21}
K.~Li, Y.~Yao, and M.~Hayashi, ``Tight exponential analysis for smoothing the
  max-relative entropy and for quantum privacy amplification,'' \emph{{IEEE}
  Transactions on Information Theory}, vol.~69, no.~3, pp. 1680--1694, mar
  2023.

\bibitem{HT17}
M.~Hayashi and V.~Y.~F. Tan, ``Equivocations, exponents, and second-order
  coding rates under various rényi information measures,'' \emph{IEEE
  Transactions on Information Theory}, vol.~63, no.~2, pp. 975--1005, February
  2017.

\bibitem{ADJ17}
A.~Anshu, V.~K. Devabathini, and R.~Jain, ``Quantum communication using
  coherent rejection sampling,'' \emph{Phys. Rev. Lett.}, vol. 119, p. 120506,
  Sep 2017.

\bibitem{CG22b}
\BIBentryALTinterwordspacing
H.-C. Cheng and L.~Gao, ``Tight one-shot analysis for convex splitting with
  applications in quantum information theory,'' 2023. [Online]. Available:
  \url{https://arxiv.org/abs/2304.12055}
\BIBentrySTDinterwordspacing

\bibitem{LY21b}
K.~Li and Y.~Yao, ``Reliable simulation of quantum channels: The error
  exponent,'' \emph{IEEE Transactions on Information Theory}, vol.~71, no.~1,
  pp. 518--529, January 2025.

\bibitem{Ume54}
H.~Umegaki, ``Conditional expectation in an operator algebra, {I},''
  \emph{Tohoku Mathematical Journal}, vol.~6, no. 2-3, pp. 177--181, jan 1954.

\bibitem{ADK+17}
A.~Anshu, V.~K. Devabathini, and R.~Jain, ``Quantum communication using
  coherent rejection sampling,'' \emph{Phys. Rev. Lett.}, vol. 119, p. 120506,
  Sep 2017.

\bibitem{AJW18}
A.~Anshu, R.~Jain, and N.~A. Warsi, ``A generalized quantum {Slepian--Wolf},''
  \emph{{IEEE} Transactions on Information Theory}, vol.~64, no.~3, pp.
  1436--1453, mar 2018.

\bibitem{BCG25}
M.~Berta, H.-C. Cheng, and L.~Gao, ``Quantum broadcast channel simulation via
  multipartite convex splitting,'' \emph{Communications in Mathematical
  Physics}, vol. 406, no.~2, January 2025.

\bibitem{LY21a}
K.~Li and Y.~Yao, ``Reliability function of quantum information decoupling via
  the sandwiched r{\'e}nyi divergence,'' \emph{Communications in Mathematical
  Physics}, vol. 405, no.~7, June 2024.

\bibitem{BSS+02}
C.~Bennett, P.~Shor, J.~Smolin, and A.~Thapliyal, ``Entanglement-assisted
  capacity of a quantum channel and the reverse {Shannon} theorem,''
  \emph{{IEEE} Transactions on Information Theory}, vol.~48, no.~10, pp.
  2637--2655, oct 2002.

\bibitem{BDH+14}
C.~H. {Bennett}, I.~{Devetak}, A.~W. {Harrow}, P.~W. {Shor}, and A.~{Winter},
  ``The quantum reverse {Shannon} theorem and resource tradeoffs for simulating
  quantum channels,'' \emph{IEEE Transactions on Information Theory}, vol.~60,
  no.~5, pp. 2926--2959, May 2014.

\bibitem{BCR11}
M.~Berta, M.~Christandl, and R.~Renner, ``The quantum reverse {Shannon} theorem
  based on one-shot information theory,'' \emph{Communications in Mathematical
  Physics}, vol. 306, no.~3, pp. 579--615, aug 2011.

\bibitem{ChristKoenRennerPostSelect}
M.~Christandl, R.~K{\"o}nig, and R.~Renner, ``Postselection technique for
  quantum channels with applications to quantum cryptography,'' \emph{Physical
  Review Letters}, vol. 102, no.~2, jan 2009.

\bibitem{CJT24}
\BIBentryALTinterwordspacing
M.~X. Cao, R.~Jain, and M.~Tomamichel, ``Quantum channel simulation in fidelity
  is no more difficult than state splitting.'' [Online]. Available:
  \url{https://arxiv.org/abs/2403.14416}
\BIBentrySTDinterwordspacing

\bibitem{Lin74}
G.~Lindblad, ``Expectations and entropy inequalities for finite quantum
  systems,'' \emph{Communications in Mathematical Physics}, vol.~39, no.~2, p.
  111–119, June 1974.

\bibitem{Petz_book_1993}
M.~Ohya and D.~Petz, \emph{Quantum Entropy and Its Use}.\hskip 1em plus 0.5em
  minus 0.4em\relax Springer Berlin, Heidelberg, 1993.

\bibitem{Uhl76}
A.~Uhlmann, ``The ``transition probability'' in the state space of a
  {$*$}-algebra,'' \emph{Reports on Mathematical Physics}, vol.~9, no.~2, pp.
  273--279, 1976.

\bibitem{GW14}
M.~K. Gupta and M.~M. Wilde, ``Multiplicativity of completely bounded
  {$p$}-norms implies a strong converse for entanglement-assisted capacity,''
  \emph{Communications in Mathematical Physics}, vol. 334, no.~2, pp. 867--887,
  October 2014.

\bibitem{CGH19}
H.-C. Cheng, L.~Gao, and M.-H. Hsieh, ``Properties of scaled noncommutative
  {R{\'{e}}nyi} and augustin information,'' in \emph{2019 {IEEE} International
  Symposium on Information Theory ({ISIT})}.\hskip 1em plus 0.5em minus
  0.4em\relax {IEEE}, jul 2019.

\bibitem{JP10}
M.~Junge and J.~Parcet, ``Mixed-norm inequalities and operator space {$L_p$}
  embedding theory,'' \emph{Memoirs of the American Mathematical Society}, vol.
  203, no. 953, 2010.

\bibitem{BL76}
J.~Bergh and J.~L{\"{o}}fstr{\"{o}}m, \emph{Interpolation Spaces}.\hskip 1em
  plus 0.5em minus 0.4em\relax Springer Berlin Heidelberg, 1976.

\end{thebibliography}
}

\end{document}